\documentclass[journal]{IEEEtran}
\usepackage{amsbsy}
\usepackage{floatflt} 
\usepackage{xcolor}

\usepackage{amsmath}
\usepackage{amssymb}
\usepackage{times}
\usepackage{graphicx}
\usepackage{xspace}
\usepackage{paralist} 
\usepackage{setspace} 
\usepackage{xypic}
\xyoption{curve}
\usepackage{latexsym}
\usepackage{theorem}
\usepackage{ifthen}
\usepackage{subfigure}
\usepackage{turnstile}
\usepackage{url}
\usepackage{booktabs}
\usepackage{subfig}
\usepackage{array}

{\theoremheaderfont{\it} \theorembodyfont{\rmfamily}
\newtheorem{theorem}{Theorem}
\newtheorem{lemma}[theorem]{Lemma}
\newtheorem{definition}[theorem]{Definition}

\newtheorem{proposition}[theorem]{Proposition}
\newtheorem{remark}[theorem]{Remark}

}

\def\ln{{\rm ln}}

\begin{document}

\title{Distributed Kalman Filtering over Massive Data Sets: Analysis Through Large Deviations of Random Riccati Equations}

\author{Di~Li,~\IEEEmembership{Student Member,~IEEE,}
        Soummya~Kar,~\IEEEmembership{Member,~IEEE,}
        Jos\'e M. F. Moura,~\IEEEmembership{Fellow,~IEEE,} \\
        H. Vincent Poor,~\IEEEmembership{Fellow,~IEEE,}
        Shuguang Cui,~\IEEEmembership{Fellow,~IEEE}
        \thanks{D. Li and S. Cui are with the Department of Electrical and Computer Engineering, Texas A\&M University, College Station, TX 77843 USA (e-mail: dili@tamu.edu; cui@tamu.edu). S. Cui is also a Distinguished Adjunct Professor at King Abdulaziz University in Saudi Arabia and a visiting professor at ShanghaiTech University, Shanghai, China. The work of D. Li and S. Cui was supported in part by DoD with grant HDTRA1-13-1-0029, by NSF with grants CNS-1343155, ECCS-1305979, and CNS-1265227, and by grant NSFC-61328102.}
        \thanks{S. Kar and J. M. F. Moura are with the Department of Electrical and Computer Engineering, Carnegie Mellon University, Pittsburgh, PA 15213 USA (e-mail: soummyak@andrew.cmu.edu; moura@ece.cmu.edu). The work of J. M. F. Moura was partially supported by NSF grants CCF-1101903 and CCF-1108509. The work of S. Kar was supported in part by NSF under grants ECCS-1306128 and ECCS-1408222}
        \thanks{H. V. Poor is with the Department of Electrical Engineering, Princeton University, Princeton, NJ 08544 USA (e-mail: poor@princeton.edu). The work of H. V. Poor was supported in part by NSF under grants DMS-1118605 and ECCS-1343210.}
       \thanks{Part of this work was presented at ICASSP'11. The correspondence author of this paper is S. Cui.} 
       \thanks{ Copyright (c) 2014 IEEE. Personal use of this material is permitted.  However, permission to use this material for any other purposes must be obtained from the IEEE by sending a request to pubs-permissions@ieee.org.}}
\maketitle

\begin{abstract}
This paper studies the convergence of the estimation error process and the characterization of the corresponding invariant measure in distributed Kalman filtering for potentially unstable and large linear dynamic systems. A gossip network protocol termed Modified Gossip Interactive Kalman Filtering (M-GIKF) is proposed, where sensors exchange their filtered states (estimates and error covariances) and propagate their observations via inter-sensor communications of rate $\overline{\gamma}$; $\overline{\gamma}$ is defined as the averaged number of inter-sensor message passages per signal evolution epoch. 
The filtered states are interpreted as stochastic particles swapped through local interaction. The paper shows that the conditional estimation error covariance sequence at each sensor under M-GIKF evolves as a random Riccati equation (RRE) with Markov modulated switching. By formulating the RRE as a random dynamical system, it is shown that the network achieves weak consensus, i.e., the conditional estimation error covariance at a randomly selected sensor converges weakly (in distribution) to a unique invariant measure. Further, it is proved that as $\overline{\gamma} \rightarrow \infty$ this invariant measure satisfies the Large Deviation (LD) upper and lower bounds, implying that this measure converges exponentially fast (in probability) to the Dirac measure $\delta_{P^*}$, where $P^*$ is the stable error covariance of the centralized (Kalman) filtering setup. The LD results answer a fundamental question on how to quantify the rate at which the distributed scheme approaches the centralized performance as the inter-sensor communication rate increases.
\end{abstract}

\section{Introduction}
\label{introduction}

\subsection{Background and Motivation}
\label{backmot}
%

For distributed estimation in a wireless sensor network \cite{Boyd-average,Olfati-topo}, multiple spatially distributed sensors collaborate to estimate the system state of interest, without the support of a central fusion center due to physical constraints such as large system size and limited communications infrastructure. Specifically, each sensor makes local partial observations and communicates with its neighbors to exchange certain information, in order to enable this collaboration. Due to its scalability for large systems and robustness to sensor failures, distributed estimation techniques find promising and wide applications including in battlefield surveillance, environment sensing, or power grid monitoring. Especially in the era of big data and large systems, which usually require overwhelming computation if implemented in centralized fashion, distributed schemes become critical since they can decompose the computational burden into local parallel procedures. A principal challenge in distributed sensing, and in distributed estimation in particular, is to design the distributed algorithm to achieve reliable and mutually agreeable estimation results across all sensors, without the help of a central fusion center. Further prior work addressing the above concerns is found in \cite{Schizas-Adhoc-consensus} and \cite{Kar-Linkfailure}, with detailed surveys in \cite{Olfati-Consensus,Dimakis-Survey} and the literature cited therein.

This paper studies the Modified Gossip Interactive Kalman Filtering (M-GIKF) for distributed estimation over potentially big data sets generated by a large dynamical system, in which each sensor observes only a portion of the large process, such that, if acting alone, no sensors can successfully resolve the entire system. The M-GIKF is fundamentally different from other distributed implementations of the Kalman filter, such as \cite{raowhyte-91}, \cite{olfati-saber:2005}, \cite{Khan-Moura}, and \cite{Schenato-Kalman}, which usually employ some form of averaging on the sensor observations/estimations through linear consensus or distributed optimization techniques. In \cite{raowhyte-91}, decentralization of the Kalman filtering algorithm is realized, where each node implements its own Kalman filter, broadcasts its estimate to every other node, and then assimilates the received information to reach certain agreement. In \cite{olfati-saber:2005}, the author proposes an approximate distributed Kalman filtering algorithm by decomposing the central Kalman filter into $n$ micro Kalman filters with inputs obtained by two consensus filters over the measurements and inverse covariance matrices. In \cite{Khan-Moura}, distributed Kalman filtering is derived for large-scale systems, where low-dimensional local Kalman filtering is achieved by spatially decomposing the large-scale system and adopting bipartite fusion graphs and consensus averaging algorithms. In \cite{Schenato-Kalman}, the authors formulate distributed Kalman filtering for a scalar system as an optimization problem to minimize the trace of the asymptotic error covariance matrix and study the interaction among the consensus matrices, the number of messages exchanged, and the Kalman gains. Single time-scale distributed approaches, i.e., in which only one round of inter-sensor message exchange is permitted per observation sampling epoch, are considered in~\cite{Khan-connectivity},~\cite{Park-augmented}. The distributed Kalman filtering algorithm in~\cite{Khan-connectivity} involves a dynamic consensus mechanism in which at every observation sampling round each sensor updates its local estimate of the system state by combining a neighborhood consensus cooperation term (based on a single round of inter-agent message exchange) with a local innovation term (based on the new observation data sensed). The resulting distributed algorithm can track unstable dynamics with bounded mean-squared error (MSE) as long as the degree of instability of the dynamics is within a so called Network Tracking Capacity (NTC) of the agent network. A generic characterization of agent networks in which the above dynamic consensus based algorithm provides tracking with bounded MSE is provided in~\cite{Khan-Topo}, where the authors employ structural system theoretic tools to obtain conditions on the communication topology and sensing model structure that guarantee tracking with bounded MSE. Another class of dynamic consensus type distributed observers/estimators has been proposed in~\cite{Park-augmented}, in which, in addition to updating their local state estimates, the agents propagate an additional \emph{augmented} state in a distributed fashion. Conditions on local innovation gain selection and coupling between the estimate and augmented state updates were obtained that guarantee stable tracking performance. More recently, an extension of the algorithm in~\cite{Khan-connectivity} is proposed in~\cite{Das-Asilomar}, which performs dynamic consensus on \emph{pseudo-innovations}, a modified version of the innovations, to improve estimation performance. A conceptually different single time-scale distributed Kalman filtering scheme was considered in~\cite{Riccati-weakcons}, in which inter-agent cooperation was obtained by randomized estimate swapping among neighboring agents. Under rather weak assumptions on the detectability of the global sensing model and connectivity of the inter-agent communication network, the algorithm in~\cite{Riccati-weakcons} was shown to yield stochastically bounded estimation error at each agent. Moreover, the conditional error covariance at each agent was shown to converge to a stationary distribution of an associated random Riccati equation.

In contrast, the proposed M-GIKF achieves sensor collaboration by exchanging local estimation states and propagating observations between neighbor sensors. In M-GIKF, each sensor runs a local Kalman filter. At each signal evolution epoch, each sensor first randomly selects a neighbor with which to exchange its state (their local Kalman filter state estimate and conditional error covariance), then propagates its observations to randomly selected neighbors, and lastly updates the estimate based on the received states and accumulated observations. This kind of collaboration through state exchange and observation propagating occurs distributedly and randomly, being controlled by the random network topology provided by an underlying gossip protocol. In M-GIKF, we assume that the communication channels among neighbors are ideal, implying that we precisely convey the sensor states and observations without distortion. The M-GIKF scheme introduced in this paper generalizes the (GIKF) scheme introduced in our prior work \cite{Riccati-weakcons}, in which inter-sensor communication and signal evolution operate at the same time scale such that only sensor states are exchanged at each signal evolution epoch; in contrast, the M-GIKF scheme is a multi-time scale algorithm in which at each signal evolution epoch the agents cooperate through a single round communication of states exchange and the additional communication at a predefined rate $\overline{\gamma}$ (informally, $\overline{\gamma}$ denotes the average number of additional network communications per signal evolution epoch) to disseminate observations according to a randomized gossip protocol.

After establishing the model for the M-GIKF, we study its conditional estimation error covariance properties. 
We show that the sensor network achieves weak consensus for each $\overline{\gamma}>0$, i.e., the conditional estimation error covariance at a randomly selected sensor converges weakly (in distribution) to a unique invariant measure of an associated random Riccati equation. To prove this, we interpret the filtered state at each sensor, including state estimate and error covariance, as a stochastic particle and interpret the travelling process of filtered states among sensors as a Markov process. In particular, the sequence of travelling states or particles evolves according to a switched system of random Riccati operators, where the switching is dictated by a nonstationary Markov chain on the network graph. We formulate the corresponding random Riccati equation (RRE) as a Random Dynamical System (RDS) and establish the asymptotic distributional properties of the RRE sequence based on the properties of RDSs, where we show that the sequence of RREs converges weakly to an invariant measure.

The GIKF proposed in our prior work \cite{Riccati-weakcons} is a simpler version of M-GIKF without observation propagation; \cite{Riccati-weakcons} shows that the error process is stochastically bounded and the network achieves weak consensus. The detailed characterization of this invariant measure was not established. In this paper, we characterize such an invariant measure denoted as $\mu^{\overline{\gamma}}$, which is the counterpart of the unique fixed point $P^*$ of the error covariance sequence in centralized Kalman filtering \cite{kalman}. As $\overline{\gamma} \rightarrow \infty$, we further prove that the measure $\mu^{\overline{\gamma}}$ approaches the Dirac measure $\delta_{P^*}$, and $\mu^{\overline{\gamma}}$ satisfies the Large Deviation (LD) upper and lower bounds. The LD property of $\mu^{\overline{\gamma}}$ implies that the probability of a rare event (the event of staying away from an arbitrary small neighborhood of $P^*$) decays exponentially; in other words, the convergence of $\mu^{\overline{\gamma}}$ to $\delta_{P^*}$ is exponentially fast in probability.

In contrast, our previous work in \cite{Riccati-moddev} only provides the Moderate Deviation property of the RRE, where the RRE arises in Kalman filtering with intermittent observations, a problem discussed in \cite{BSinopoli-intermittent}, where the sensor observation packets, transmitted through an imperfect communication medium, are received at the estimator as a Bernoulli process with arrival probability $\gamma > 0$. In this case, the Moderate Derivation shows that the probability of a rare event decays as a power law of $(1-\gamma)$ for $\gamma \rightarrow 1$. Such setup and result are fundamentally different from those in this paper, because Kalman filtering with intermittent observations discussed in \cite{BSinopoli-intermittent} and \cite{Riccati-moddev} considers only the local algorithm at each sensor without inter-sensor communications.

The rest of this paper is organized as follows. Section~\ref{notprel} presents the notation and preliminaries. Section~\ref{prob_form} sets up the problem, reviews the previous GIKF, introduces the M-GIKF, gives an example of the distributed observation dissemination protocol, and establishes an interactive particle interpretation and the RDS formulation of the switching iterates of the RRE. Section~\ref{main_res} presents our main results regarding the convergence of the M-GIKF and the LD property for the resulting invariant measure. Section~\ref{app_res} provides some intermediate results for the Riccati equation.  Sections~\ref{ld-inv} and~\ref{proof_th:ldp} discuss the steps to prove the LD property. Section~\ref{simulation} presents the simulation results and Section~\ref{conclusion} concludes the paper.


\subsection{Notation and Preliminaries}
\label{notprel}
Denote by: $\mathbb{R}$, the reals; $\mathbb{R}^{M}$, the $M$-dimensional Euclidean space; $\mathbb{T}$, the integers; $\mathbb{T}_{+}$, the non-negative integers; $\mathbb{N}$, the natural numbers; and $\mathcal{X}$, a generic space.
For a subset $B\subset \mathcal{X}$, $\mathbb{I}_{B}:\mathcal{X}\longmapsto\{0,1\}$ is the indicator
function, which is~$1$ when the argument is in~$B$ and zero otherwise; and $id_{\mathcal{X}}$ is the identity function on
$\mathcal{X}$. For a set $\Gamma \subset \mathcal{X}$, we denote by $\Gamma^\circ$ and $\overline{\Gamma}$ its interior and closure, respectively. For $x \in \mathcal{X}$, the open ball of radius $\varepsilon > 0$ centered at $x$ is denoted by $B_{\varepsilon}(x)$.


\textbf{Cones in partially ordered Banach spaces and probability measures on metric spaces:}
For this part of notation and definitions, we refer the reader to Section I in \cite{Riccati-weakcons}.

\textbf{Limit:}
Let $f:\mathbb{R}\longmapsto\mathbb{R}$ be a measurable function. The notation $\lim_{z\rightarrow x}f(z)=y$ implies that, for every sequence $\{z_{n}\}_{n\in\mathbb{N}}$ in $\mathbb{R}$ with $\lim_{n\rightarrow\infty}|z_{n}-x|=0$, we have $\lim_{n\rightarrow\infty}|f(z_{n})-y|=0$.
The notation $\lim_{z\uparrow x}f(z)=y$ implies that for every sequence $\{z_{n}\}_{n\in\mathbb{N}}$ in $\mathbb{R}$ with $z_{n}<x$ and $\lim_{n\rightarrow\infty}|z_{n}-x|=0$, we have $\lim_{n\rightarrow\infty}|f(z_{n})-y|=0$.
\textbf{Large Deviations:}
\label{moddev100}
Let $\left\{\mathbb{\mu}^{\overline{\gamma}}\right\}$ be a family of probability measures on the complete separable metric space $(\mathcal{X},d_{\mathcal{X}})$ indexed by the real-valued parameter $\overline{\gamma}$ taking values in $\mathbb{R}_{+}$.
Let $\overline{I}:\mathcal{X}\longmapsto\overline{\mathbb{R}}_{+}$ be an extended-valued lower semicontinuous function.
The family $\left\{\mathbb{\mu}^{\overline{\gamma}}\right\}$ is said to satisfy a large deviations upper bound with rate function $\overline {I} (\cdot)$ if the following holds:
\begin{equation}
\label{moddev2}
\limsup_{\overline{\gamma}\rightarrow\infty}\frac{1}{\overline{\gamma}}\ln\mathbb{\mu}^{\overline{\gamma}}\!\left(\mathcal{F}\right)\!\leq\! -\!\inf_{X\in\mathcal{F}}\overline{I}(X),~\mbox{for every closed set $\mathcal{F}\!\in\!\mathcal{X}$}.
\end{equation}
Similarly, for an extended-valued lower semicontinuous function $\underline{I}:\mathcal{X}\longmapsto\overline{\mathbb{R}}_{+}$, the family $\left\{\mathbb{\mu}^{\overline{\gamma}}\right\}$ is said to satisfy a large deviations lower bound with rate function $\underline{I}(\cdot)$, if
\begin{equation}
\label{moddev1}
\liminf_{\overline{\gamma}\rightarrow\infty}\frac{1}{\overline{\gamma}}\ln\mathbb{\mu}^{\overline{\gamma}}\left(\mathcal{O}\right)\!\geq\! -\inf_{X\!\in\!\mathcal{O}}\underline{I}(X),~\mbox{for every open set $\mathcal{O}\in\mathcal{X}$}.
\end{equation}
In addition, if the functions $\overline{I}$ and $\underline{I}$ coincide, i.e., $\overline{I}=\underline{I}=I$, the family $\left\{\mathbb{\mu}^{\overline{\gamma}}\right\}$ is said to satisfy a large deviations principle (LDP) with rate function $I(\cdot)$ (see~\cite{DeuschelStroock}).
The lower semicontinuity implies that the level sets of $\overline{I}(\cdot)$ (or $\underline{I}(\cdot)$), i.e., sets of the form $\{X\in\mathcal{X}~|~\overline{I}(X)\leq\alpha\}$ (or $\underline{I}(\cdot)$) for every $\alpha\in\mathbb{R}_{+}$, are closed. If, in addition, the levels sets are compact (for every $\alpha$), $\overline{I}(\cdot)$ (or $\underline{I}(\cdot)$) is said to be a good rate function.

Before interpreting the consequences of the LD upper and lower bounds as defined above, we consider the notion of a rare event, which is the central motivation to all large deviations:
\begin{definition}[Rare Event]
\label{moddev3} A set $\Gamma\subset\mathcal{B}(\mathcal{X})$ is called a rare event with respect to (w.r.t.) the family $\left\{\mathbb{\mu}^{\overline{\gamma}}\right\}$ of probability measures, if $\lim_{\overline{\gamma}\rightarrow\infty}\mathbb{\mu}^{\overline{\gamma}}(\Gamma)=0$.
In other words, the event $\Gamma$ becomes increasingly difficult to observe (i.e., it becomes rare) as $\overline{\gamma}\rightarrow\infty$.
\end{definition}
Once a rare event $\Gamma$ is identified, the next natural question is the rate at which its probability goes to zero under $\mathbb{\mu}^{\overline{\gamma}}$ as $\overline{\gamma}\rightarrow\infty$. This is answered by the LD upper and lower bounds, which also characterize the family $\left\{\mathbb{\mu}^{\overline{\gamma}}\right\}$ as $\overline{\gamma}\rightarrow\infty$. Indeed, it is not hard to see that, if the family $\left\{\mathbb{\mu}^{\overline{\gamma}}\right\}$ satisfies the LD upper and lower bounds, we have for every measurable set $\Gamma\in\mathcal{X}$:
\begin{equation}
\label{moddev5}
\mathbb{\mu}^{\overline{\gamma}}(\Gamma)\leq e^{-\overline{\gamma}(\inf_{X\in\overline{\Gamma}}\overline{I}(X)+o(1))}
\end{equation}
\begin{equation}
\label{moddev500}
\mathbb{\mu}^{\overline{\gamma}}(\Gamma)\geq e^{-\overline{\gamma}(\inf_{X\in\Gamma^{\circ}}\underline{I}(X)+o(1))},
\end{equation}
where $o(1)$ is the little-$o$ notation. Now assume $\inf_{X\in\overline{\Gamma}}\overline{I}(X)>0$. Then, from~\eqref{moddev5} it is clear that $\Gamma$ is a rare event and, in fact, we conclude that the probability of $\Gamma$ decays exponentially with a LD exponent greater than or equal to $\inf_{X\in\overline{\Gamma}}\overline{I}(X)$. Similarly, $\inf_{X\in\Gamma^{\circ}}\underline{I}(X)>0$ suggests that the LD decay exponent is not arbitrary and cannot be larger than $\inf_{X\in\Gamma^{\circ}}\underline{I}(X)$. In addition, if the rate functions $\overline{I}$ and $\underline{I}$ close to each other, the estimate of the exact decay exponent is tight.

We summarize the key symbols used in this paper in Table~\ref{tab}.
\begin{table}[tbp]\caption{Table of Key Symbols}\label{tab}
\renewcommand{\arraystretch}{1.25}
\centering
\begin{tabular}{c|p{6cm}<{\centering}}
\hline
$\mathcal{F}$ & dynamical system  matrix \\
\hline
$\mathcal{Q}$ & system noise covariance\\
\hline
$\mathcal{C}_{n}$ &observation matrix at sensor $n$ \\
\hline
$\mathcal{R}_{n}$ & observation noise covariance at sensor $n$ \\
\hline
$\overline{A}$ & transition matrix of the Markov chain \\
\hline
$\mathfrak{P}$ & power set of sensor index $[1,\cdots,N]$\\
\hline
$\jmath$& index of elements of $\mathfrak{P}$\\
\hline
${\mathcal{I}^{n}_{k}}$& index of sensors whose observations are available at sensor $n$\\
\hline
$\overline{\gamma}$& averaged number of inter-sensor message passages per signal evolution epoch\\
\hline
$\{\widehat{P}_{k}^{n}\}$ & error covariance sequence at sensor $n$ defined in (12)\\
\hline
$\overline{q}_{n}(\jmath)$  & upper bound defined in (16)\\
\hline
$\underline{q}_{n}(\jmath)$ & lower bound defined in (16)\\
\hline
$\{P_n(k)\}$ & sequence of switched Riccati iterates defined in (29)\\
\hline
$\{\widetilde{P}(k)\}$ & an auxiliary process defined in (34)\\
\hline
$\overline{w}(\mathcal{R})$ & upper weight defined in (42) \\
\hline
$\underline{w}(\mathcal{R})$ & lower weight defined in (43)\\
\hline
$\overline{I}(\cdot)$& upper large deviation rate function \\
\hline
$\underline{I}(\cdot)$ &lower large deviation rate function\\
\hline
\end{tabular}
\end{table}

\section{Distributed Kalman Filtering: Algorithms and Assumptions}
\label{prob_form}
In this section, we present a generic class of distributed algorithms for Kalman filtering. The problem setup is described in Section~\ref{setup}. In Section~\ref{GIKF-prior}, the algorithm GIKF (Gossip Interactive Kalman Filter) given in~\cite{Riccati-weakcons} is briefly reviewed. The GIKF is generalized for cases where additional inter-sensor communication is available in Section~\ref{M-GIKF}. An example of the observation dissemination protocol is given in Section~\ref{ex-dist}.
To facilitate later analysis, an interacting particle representation of the random Riccati equation resulting from the M-GIKF is provided in Section~\ref{int_part_rep}.

\subsection{Problem Setup}
\label{setup}
The system model and communication model are adopted from that in our prior work \cite{Riccati-weakcons}, which is included here for completeness.

\textbf{Signal/Observation Model:}
Let $t\in\mathbb{R}_{+}$ denote continuous time and $\Delta>0$ be a constant sampling interval. The global unknown signal process $\{\mathbf{x}_{k\Delta}\}_{k\in\mathbb{N}}$ evolves as a sampled linear dynamical system:
\begin{equation}
\label{sys_model}
\mathbf{x}_{(k+1)\Delta}=\mathcal{F}\mathbf{x}_{k\Delta}+\mathbf{w}_{k\Delta}
\end{equation}
where $\mathbf{x}_{k\Delta}\in\mathbb{R}^{M}$ is the signal (state)
vector with an initial state $\mathbf{x}_{0}$ distributed as a
zero mean Gaussian vector with covariance $\widehat{P}_{0}$ and
the system noise $\{\mathbf{w}_{k\Delta}\}$ is an uncorrelated zero mean
Gaussian sequence independent of $\mathbf{x}_{0}$ with covariance
$\mathcal{Q}$. The observation at the $n$-th sensor
$\mathbf{y}^{n}_{k\Delta}\in\mathbb{R}^{m_{n}}$ at time $k\Delta$ is of the
form:
\begin{equation}
\label{obs_n}
\mathbf{y}^{n}_{k\Delta}=\mathcal{C}_{n}\mathbf{x}_{k\Delta}+\mathbf{v}^{n}_{k\Delta}
\end{equation}
where $\mathcal{C}_{n}\in\mathbb{R}^{m_{n}\times M}$ is the observation matrix and
$\{\mathbf{v}^{n}_{k\Delta}\}$ is an uncorrelated zero mean Gaussian
observation noise sequence with covariance $\mathcal{R}_{n}\gg
\mathbf{0}$\footnote{The sampling interval $\Delta$ could be a function of various system parameters such as the sampling rate of the sensors and the rate of signal evolution. Thus the factor $1/\Delta$ may be viewed as the signal evolution rate. Since $\Delta$ is fixed throughout the paper, we will drop $\Delta$ from the discrete index of sampled processes for notational convenience. Then, $\mathbf{x}_{k}$ will be used to denote $\mathbf{x}_{k\Delta}$ and the process $\{\mathbf{x}_{k\Delta}\}_{k\in\mathbb{N}}$ will be denoted by $\{\mathbf{x}_{k}\}_{k\in\mathbb{N}}$.}. Also, the noise sequences at different sensors are
independent of each other, of the system noise process, and of the
initial system state. Due to the limited capability of each individual sensor, typically the dimension of $\mathbf{y}^{n}_{k}$ is much
smaller than that of the signal process. Thus the observation
process at each sensor is usually not sufficient to make the pair
$\{\mathcal{C}_{n},\mathcal{F}\}$ observable\footnote{It is possible that some of the sensors have no observation
capabilities, i.e., the corresponding $C_{n}$ is a zero matrix.
Thus this formulation easily carries over to networks of
heterogeneous agents, which consist of `sensors' actually sensing
the field of interest and some pure actuators implementing local control
actions based on the estimated field.}. We envision a totally
distributed application where a reliable estimate of the entire signal
process is required at each sensor. To achieve this, the sensors need collaboration via occasional communications with their neighbors, whereby
they exchange their observations and filtering states. The details of the collaboration scheme will be defined precisely later.

\textbf{Communication Model:}
Communication among sensors is constrained by several factors such as proximity, transmit power, and receiving capabilities. We model the underlying communication structure of the network in terms of an undirected graph $(V,\mathcal{E})$, where $V$ denotes the set of $N$ sensors and $\mathcal{E}$ is the set of edges or allowable communication links between the sensors. The notation $n\sim l$ indicates that sensors $n$ and $l$ can communicate, i.e., $\mathcal{E}$ contains the undirected edge $(n,l)$. The graph $(V,\mathcal{E})$ is represented by the $N\times N$ symmetric maximal adjacency matrix $\mathcal{A}$:
\begin{equation}
\label{def_mathcalA}
\mathcal{A}_{nl}=\left\{ \begin{array}{ll}
                    1, & \mbox{if $(n,l)\in \mathcal{E}$}\\
                    0, & \mbox{otherwise}.
                   \end{array}
          \right.
\end{equation}
We assume that the diagonal elements of $\mathcal{A}$ are
identically~1, indicating that a sensor $n$ can always communicate
with itself. Note that $\mathcal{E}$ is the maximal set
of allowable communication links in the network at any time; however, at a particular
instant, each sensor may choose to communicate only to a fraction
of its neighbors. The exact communication protocol is not so
important for the analysis, as long as some weak
connectivity assumptions are satisfied. For definiteness, we
assume the following generic communication model, which subsumes
the widely used gossip protocol for real time embedded
architectures \cite{Boyd-GossipInfTheory} and the graph matching
based communication protocols for internet architectures
\cite{Mckeown}. 
 Define the set $\mathcal{M}$
of symmetric 0-1 $N\times N$ matrices as follows:
\begin{equation}
\label{def_mathcalM}\mathcal{M}=\left\{A~\left|~\mathbf{1}^{T}A=\mathbf{1}^{T},~~A\mathbf{1}=\mathbf{1},~~A\leq \mathcal{E}\right.\right\}
\end{equation}
where $A\leq \mathcal{E}$ is to be interpreted component-wise. In other words, $\mathcal{M}$ is the set of adjacency matrices,
where every node is incident to exactly one edge (including the self edge) and allowable edges are only those included in
$\mathcal{E}$.\footnote{The set $\mathcal{M}$ is non-empty, since the $N\times N$ identity matrix $I_{N}\in\mathcal{M}$.}
Let $\mathcal{D}$ be a probability distribution on the space
$\mathcal{M}$. The sequence of time-varying adjacency matrices
$\{A(k)\}_{k\in\mathbb{N}}$, governing the inter-sensor
communication, is then an independent and identically distributed (i.i.d.) sequence in $\mathcal{M}$ with
distribution $\mathcal{D}$ and independent of the signal and
observation processes.\footnote{For convenience of presentation,
we assume that $A(0)=I_{N}$, while real communication starts at time slot
$k=1$.} We make the following assumption of connectivity on average.

\textbf{Assumption C.1}: Define the symmetric stochastic matrix
\label{assumptionc.1}
$\overline{A}$ as
\begin{equation}
\label{def_barA}\overline{A}=\mathbb{E}\left[A(k)\right]=\int_{\mathcal{M}}A\,d\,\mathcal{D}(A).
\end{equation}
The matrix $\overline{A}$ is assumed to be irreducible and aperiodic.

\begin{remark}
\label{rem1} The stochasticity of $\overline{A}$ is inherited from
that of the elements in $\mathcal{M}$. Here, we are not concerned
with the properties of the distribution $\mathcal{D}$ as long as
the weak connectivity assumption above is satisfied. The irreducibility of $\overline{A}$ depends both on the set
of allowable edges $\mathcal{E}$ and the distribution
$\mathcal{D}$. We do not detail this question here.
However, to show the applicability of Assumption~\textbf{C.1} and
justify the notion of weak connectivity, we note that such a
distribution $\mathcal{D}$ always exists if the graph
$(V,\mathcal{E})$ is connected. 
We provide a Markov chain interpretation of the mean adjacency matrix $\overline{A}$, which is helpful for the following analysis. The matrix
$\overline{A}$ can be interpreted as the transition matrix of a
time-homogeneous Markov chain on the state space $V$. Since the
state space $V$ is finite, the irreducibility of $\overline{A}$ implies the positive recurrence of the resulting Markov chain.
\end{remark}

We present the following weak assumptions on the signal/observation model.

\textbf{Stabilizability: Assumption S.1} The pair
\label{assumptions.1}
$(\mathcal{F},\mathcal{Q}^{1/2})$ is stabilizable. The
non-degeneracy (positive definiteness) of $\mathcal{Q}$ guarantees this.

\textbf{Weak Detectability: Assumption D.1}  There exists a walk\footnote{A walk is defined w.r.t. the graph induced by the non-zero entries of the matrix $\overline{A}$.} $(n_1,\cdots,n_l)$ of length $l\geq 1$ covering the $N$ nodes, such that the matrix $\sum_{i=1}^{l}(\mathcal{F}^{i-1})^T\mathcal{C}_{n_i}^T \mathcal{C}_{n_i} \mathcal{F}^{i-1}$ is invertible.

\begin{remark}
%
Assumption (D.1) is minimal, since, for an arbitrary choice of the matrix $\mathcal{F}$ governing the signal dynamics, even a centralized setting, where a center accesses all the sensor observations over all time, requires an equivalent detectability condition. This justifies the term weak detectability.
\end{remark}

\subsection{Prior Work: Algorithm GIKF}
\label{GIKF-prior}
 The GIKF (see~\cite{Riccati-weakcons}) assumes that the inter-sensor communication rate is comparable to the signal evolution rate and only one round of sensor communication is allowed for every epoch $[(k-1)\Delta,k\Delta)$.

We now present the algorithm GIKF (gossip based interacting Kalman
filter) for distributed estimation of the signal process
$\mathbf{x}_{k}$ over time. Let the filter at sensor $n$ be initialized with the pair
$\left(\widehat{\mathbf{x}}_{0|-1},\widehat{P}_{0}\right)$, where
$\widehat{\mathbf{x}}_{0|-1}$ denotes the prior estimate of
$\mathbf{x}_{0}$ (with no observation information) and
$\widehat{P}_{0}$ the corresponding error covariance. Also, $(\widehat{\mathbf{x}}_{k|k-1}^{n},\widehat{P}_{k}^{n})$ denotes
the estimate at sensor $n$ of $\mathbf{x}_{k}$ based on
information\footnote{The information at sensor $n$ till (and
including) time $k$ corresponds to the sequence of observations
$\{\mathbf{y}^{n}_{s}\}_{0\leq s\leq k}$ obtained at the sensor
and the information received by data exchange with its
neighboring sensors.} till time $k-1$ and the corresponding
conditional error covariance, respectively. The pair
$\left(\widehat{\mathbf{x}}_{k|k-1}^{n},\widehat{P}_{k}^{n}\right)$ is also
referred to as the state of sensor $n$ at time $k-1$. To define
the estimate update rule for the GIKF, denote by $n_k^{\rightarrow}$ the neighbor of sensor $n$ at time $k$ w.r.t.~the adjacency matrix\footnote{Note that by symmetry we have $(n_k^{\rightarrow})_k^{\rightarrow}=n$. It is possible that
$n_k^{\rightarrow}=n$, in which case $A(k)$ has a self-loop at node $n$.} $A(k)$. We assume that all inter-sensor communication for time $k$ occurs at the beginning of the slot, whereby communicating sensors swap their previous states, i.e., if at time $k$, $n_k^{\rightarrow}=l$, sensor $n$ replaces its previous state $\left(\widehat{\mathbf{x}}_{k|k-1}^{n},\widehat{P}_{k}^{n}\right)$ by $\left(\widehat{\mathbf{x}}_{k|k-1}^{l},\widehat{P}_{k}^{l}\right)$ and sensor $l$ replaces its previous state $\left(\widehat{\mathbf{x}}_{k|k-1}^{l},\widehat{P}_{k}^{l}\right)$ by $\left(\widehat{\mathbf{x}}_{k|k-1}^{n},\widehat{P}_{k}^{n}\right)$. The estimate is updated by sensor $n$ at the end of the slot (after the communication and observation tasks have been completed) as follows:
\begin{eqnarray}
\label{est_up}
\widehat{\mathbf{x}}_{k+1|k}^{n}&=&\mathbb{E}\left[\mathbf{x}_{k+1}~\left|~\widehat{\mathbf{x}}_{k|k-1}^{n_k^{\rightarrow}},
\widehat{P}_{k}^{n_k^{\rightarrow}},\mathbf{y}^{n}_{k}\right.\right]\\
\label{est_up1}
\widehat{P}_{k+1}^{n}&=&\mathbb{E}\left[\left(\mathbf{x}_{k+1}-\widehat{\mathbf{x}}_{k+1|k}^{n}\right)
\left(\mathbf{x}_{k+1}-\widehat{\mathbf{x}}_{k+1|k}^{n}\right)^{T}\right.\nonumber \\ & & \left.~~~~\left|~\widehat{\mathbf{x}}_{k|k-1}^{n_k^{\rightarrow}},
\widehat{P}_{k}^{n_k^{\rightarrow}},\mathbf{y}^{n}_{k}\right.\right].
\end{eqnarray}


\subsection{Faster Communication Time-scale: Algorithm M-GIKF}
\label{M-GIKF} 
We recall from Table I that $\mathfrak{P}$ denotes the power set of $[1,\cdots,N]$ and the elements of $\mathfrak{P}$ are indexed by $\jmath\in [0,\cdots, 2^{N}-1]$, with 0 denoting the null set and $2^{N}-1$ the entire set. Also, for technical convenience, we will interpret the elements (sensors) in a subset index by $\jmath$ to be arranged in ascending order, $i_{1}$ denoting the first and $i_{|\mathfrak{P}_\jmath|}$ denoting the last, i.e., this subset is $\{i_{1},\cdots,i_{|\mathfrak{P}_\jmath|}\}$. 
For a given $\jmath\in[0,\cdots, 2^{N}-1]$, we denote the subset $[(\mathbf{y}_{k}^{i_{1}})^{T}\cdots(\mathbf{y}_{k}^{i_{|\mathfrak{P}_\jmath|}})^{T}]^{T}$ of observations at the $k$-th epoch by $\mathbf{y}_{k}^{\jmath}$, whereas the matrices $\mathcal{C}_{\jmath}$ and $\mathcal{R}_{\jmath}$ are $\mathcal{C}_{\jmath}=[\mathcal{C}_{i_{1}}^{T}\cdots\mathcal{C}_{i_{|\mathfrak{P}_\jmath|}}^{T}]^{T}$ and $\mathcal{R}_{\jmath}=\mbox{diag}[\mathcal{R}_{i_{1}},\cdots,\mathcal{R}_{i_{|\mathfrak{P}_\jmath|}}]$.

Suppose, in the basic GIKF scheme explained above, there is an additional step of communication. Specifically, assume that in every interval $[k\Delta,(k+1)\Delta)$ the network (as a whole) is given an opportunity for additional communication at rate $\overline{\gamma}$, i.e., additional $\overline{\gamma}$ message exchanges occur across the network in each epoch. In particular, we assume that the total number of additional sensor transmissions in $[k\Delta,(k+1)\Delta)$ is dominated by a Poisson random variable of rate $\overline{\gamma}$\footnote{The Poisson assumption is claimed and justified at the end of this subsection.}, and that each transmission conforms to the network topology induced by the maximal adjacency matrix $\mathcal{A}$. Clearly, by exploiting this additional inter-sensor communication, the network should be able to perform a filtering task that is at least as good if not better than the basic GIKF.

A natural way to improve the performance of the GIKF is to use this additional communication to disseminate the observations across the sensors. We denote this new scheme with additional communication for disseminating the observations by Modified GIKF (M-GIKF). For each sensor $n$, the subset-valued process $\{\mathcal{I}^{\overline{\gamma},n}_{k}\}$ taking values in $\mathfrak{P}$ is used to index the subset of observations ${\mathbf{y}^{\mathcal{I}^{\overline{\gamma},n}_{k}}_{k}}$ available at sensor $n$ at the end of the interval $[k\Delta,(k+1)\Delta)$, e.g., if $\{\mathcal{I}^{\overline{\gamma},n}_{k}\}=[m,n]$, then the observations $\mathbf{y}^m_k$ and $\mathbf{y}^n_k$ are available at sensor $n$ by the end of the interval $[k\Delta,(k+1)\Delta)$. Also, the corresponding parameters with ${\mathbf{y}^{\mathcal{I}^{\overline{\gamma},n}_{k}}_{k}}$ in the observation model \eqref{obs_n} are denoted by $\mathcal{C}_{\mathcal{I}^{\overline{\gamma},n}_{k}}$ and $\mathcal{R}_{\mathcal{I}^{\overline{\gamma},n}_{k}}$.

For the GIKF algorithm it is clear that
\[
\mathcal{I}^{0,n}_{k}=\{n\},~\forall n\in [1,\cdots,N],~k\in\mathbb{N},
\]
i.e., each sensor only has access to its own observations in each epoch. Hence, in the GIKF the only cooperation among the sensors is achieved through estimate exchanging and no explicit mixing or aggregation of instantaneous observations occur. This is in fact the key difference between the GIKF and the M-GIKF. In the M-GIKF, the sensors use the additional communication rate $\overline{\gamma}$ to exchange instantaneous observations, in addition to performing the basic estimate swapping of the GIKF.

In this work, our main focus is not on the exact nature of the instantaneous observation dissemination protocol, as long as it is distributed (i.e., any inter-sensor exchange conforms to the network topology) and satisfies some assumptions (in Section~\ref{ex-dist} we will provide an example of such distributed protocols satisfying these assumptions). Recall $\mathcal{I}^{\overline{\gamma},n}_{k}$ to be the instantaneous observation set available at sensor $n$ by the end of the interval $[k\Delta,(k+1)\Delta)$. Note, the statistics of the process $\{\mathcal{I}^{\overline{\gamma},n}_{k}\}$ depend on the dissemination protocol used and the operating rate $\overline{\gamma}$. Before providing details of the dissemination protocol and the assumptions on the processes $\{\mathcal{I}^{\overline{\gamma},n}_{k}\}$, for all $n$, we explain the M-GIKF scheme as follows. For the moment, the reader may assume that $\{\mathcal{I}^{\overline{\gamma},n}_{k}\}$ are generic set-valued processes taking values in $\mathfrak{P}$ and there exists a distributed protocol operating in the time window $[k\Delta,(k+1)\Delta)$ leading to such observation sets at the sensors by the end of the epoch. Clearly, for any protocol and $\overline{\gamma}\geq 0$\footnote{For conciseness, we will drop the superscript $\overline{\gamma}$ over the notations related to the M-GIKF with the additional communication rate $\overline{\gamma}$.},
\[
\{n\}\subset\mathcal{I}^{\overline{\gamma},n}_{k},~\forall n\in [1,\cdots,N],~k\in\mathbb{N}.
\]
Moreover, if the observation dissemination protocol is reasonable, $\mathcal{I}^{n}_{k}$ is strictly greater than $\{n\}$ with positive probability. The basic difference between the GIKF and the M-GIKF is that, in~\eqref{est_up}-\eqref{est_up1}, instead of conditioning on $\mathbf{y}_{k}^{n}$ at sensor $n$, we condition on the possibly larger set $\mathbf{y}^{\mathcal{I}_{k}^{n}}_{k}$ of observations available at sensor $n$.

With this setup, now, we formally describe the M-GIKF, which generalizes the GIKF when additional inter-sensor communication at rate $\overline{\gamma}$ is allowed in every epoch $[k\Delta,(k+1)\Delta)$.

\textbf{Algorithm M-GIKF}:
We assume that $\overline{\gamma}>0$ is given and fixed. Let the filter at sensor $n$ be initialized with the pair $\left(\widehat{\mathbf{x}}_{0|-1},\widehat{P}_{0}\right)$, where $\widehat{\mathbf{x}}_{0|-1}$ denotes the prior estimate of $\mathbf{x}_{0}$ (with no observation information) and $\widehat{P}_{0}$ the corresponding error covariance. Also, by $\left(\widehat{\mathbf{x}}_{k|k-1}^{n},\widehat{P}_{k}^{n}\right)$ we denote the estimate at sensor $n$ of $\mathbf{x}_{k}$ based on
information till time $k-1$ and the corresponding conditional error covariance, respectively. The pair $\left(\widehat{\mathbf{x}}_{k|k-1}^{n},\widehat{P}_{k}^{n}\right)$ is also referred to as the state of sensor $n$ at time $k-1$. Similar to GIKF, the M-GIKF update involves the state exchanging step (w.r.t. the adjacency matrices $\{A(k)\}$), whereby, at the beginning of the epoch $[k\Delta,(k+1)\Delta)$, sensor $n$ exchanges its state with its neighbor $n_k^{\rightarrow}$ w.r.t. $A(k)$. This exchange is performed only once in the interval $[k\Delta,(k+1)\Delta)$. Then each sensor in M-GIKF makes its sensing observation and M-GIKF instantiates the distributed dissemination protocol before the end of the epoch. This leads observation aggregation with $\mathbf{y}^{\mathcal{I}^{n}_{k}}_{k}$ being the observation set available at sensor $n$ at the end of the interval $[k\Delta,(k+1)\Delta)$. The estimate update at sensor $n$ at the end of the slot (after the communication and observation dissemination tasks have been completed) is
\begin{eqnarray}
\widehat{\mathbf{x}}_{k+1|k}^{n}&=&\mathbb{E}\left[\mathbf{x}_{k+1}~\left|~\widehat{\mathbf{x}}_{k|k-1}^{n_k^{\rightarrow}},
\widehat{P}_{k}^{n_k^{\rightarrow}},\mathbf{y}^{\mathcal{I}^{n}_{k}}_{k},\mathcal{I}^{n}_{k}\right.\right]\nonumber\\
\widehat{P}_{k+1}^{n}&=&\mathbb{E}\left[\left(\mathbf{x}_{k+1}-\widehat{\mathbf{x}}_{k+1|k}^{n}\right)
\left(\mathbf{x}_{k+1}-\widehat{\mathbf{x}}_{k+1|k}^{n}\right)^{T}~\right.\nonumber \\ & & \left.~~~~\left|~\widehat{\mathbf{x}}_{k|k-1}^{n_k^{\rightarrow}},
\widehat{P}_{k}^{n_k^{\rightarrow}},\mathbf{y}^{\mathcal{I}^{n}_{k}}_{k},\mathcal{I}^{n}_{k}\right.\right].\nonumber
\end{eqnarray}

\noindent  Due to conditional Gaussianity, the optimal prediction steps can be implemented through the time-varying Kalman filter recursions, and it follows that the sequence $\left\{\widehat{P}_{k}^{n}\right\}$ of the conditional predicted error covariance matrices at sensor $n$ satisfies the Riccati recursion:
\begin{align}
\label{est_up2}
&\widehat{P}_{k+1}^{n}=\mathcal{F}\widehat{P}_{k}^{n_k^{\rightarrow}}\mathcal{F}^{T}+
\mathcal{Q}-\mathcal{F}\widehat{P}_{k}^{n_k^{\rightarrow}}\mathcal{C}_{\mathcal{I}^{n}_{k}}^{T}\nonumber\\
&\times\left(\mathcal{C}_{\mathcal{I}^{n}_{k}}\widehat{P}_{k}^{n_k^{\rightarrow}}\mathcal{C}_{\mathcal{I}^{n}_{k}}^{T}+
\mathcal{R}_{\mathcal{I}^{n}_{k}}\right)^{-1}\mathcal{C}_{\mathcal{I}^{n}_{k}}\widehat{P}_{k}^{n_k^{\rightarrow}}\mathcal{F}^{T}.
\end{align}
\begin{remark}
\label{rem:233}
Note that the sequence $\left\{\widehat{P}_{k}^{n}\right\}$ is random, due to the random neighborhood selection function $n_k^{\rightarrow}$. The goal of the paper is to study the asymptotic properties of the sequence of random conditional error covariance matrices $\left\{\widehat{P}_{k}^{n}\right\}$ at each sensor $n$ and to show in what sense they reach consensus, such that, in the limit of large time, every sensor provides an equally good (stable in the sense of estimation error) estimate of the signal process.
\end{remark}
\textbf{Assumptions on the observation dissemination protocol:}
We introduce the following assumptions on the communication medium and the distributed information dissemination protocol generating the subsets
$\{\mathcal{I}^{n}_{k}\}$ for all $n,k$.
\begin{itemize}
\item[(i)]~\textbf{(E.1)}: The total number of inter-sensor observation dissemination messages $\mathcal{M}(k)$ in the interval $[k\Delta,(k+1)\Delta)$, for all $k\in\mathbb{T}_{+}$ follows a Poisson distribution with mean $\overline{\gamma}$.
\item[(ii)]~\textbf{(E.2)}: For each $n$, the process $\{\mathcal{I}^{n}_{k}\}$ is (conditionally) i.i.d. For each $k$, the protocol initiates at the beginning of the interval $[k\Delta,(k+1)\Delta)$ and operates on the most recent observations $\{\mathbf{y}^{n}_{k}\}_{1\leq n\leq N}$. The protocol terminates at the end of the epoch. For observation dissemination in the next epoch $[(k+1)\Delta,(k+2)\Delta)$, the protocol is re-initiated and acts on the new observation data $\{\mathbf{y}^{n}_{k+1}\}_{1\leq n\leq N}$, independent of its status in the previous epoch. Necessarily, the sequence is (conditionally) i.i.d.. We define
    \begin{equation}
    \label{ass-obs1}
    \lim_{k\rightarrow\infty}\frac{1}{k}\sum_{i=0}^{k-1}\mathcal{M}(i)=\overline{\gamma},~\mbox{a.s.},
    \end{equation}
    i.e., the average number of dissemination messages per epoch is $\overline{\gamma}$.
\item[(iii)]~\textbf{(E.3)}: Recall the notations $\jmath$ and $\{i_1,\cdots,i_{|\mathfrak{P}_\jmath|}\}$ at the beginning of this section. For each $\jmath\in[0,\cdots, 2^{N}-1]$, define
\begin{equation}
\label{ass-obs4}
\mathbb{P}\left(\mathcal{I}^{n}_{k}=\{i_1,\cdots,i_{|\mathfrak{P}_\jmath|}\}\right)=q_{n}(\jmath),~\forall n,k.
\end{equation}
We assume that for all $\overline{\gamma}>0$
\begin{equation}
\label{ass-obs2}
\mathbb{P}\left(\mathcal{I}^{n}_{k}=\{1,2,\cdots,N\}\right)=q_{n}(2^{N}-1)>0,~\forall n,k.
\end{equation}
\item[(iv)]~\textbf{(E.4)}: For each $\jmath\neq 2^{N}-1$, define
\begin{align}\label{E4_protocol}
-\underline{q}_{n}(\jmath)\leq\liminf_{\overline{\gamma}\rightarrow\infty}&\frac{1}{\overline{\gamma}}\ln\left(q_{n}(\jmath)\right)\nonumber\\
&\leq\limsup_{\overline{\gamma}\rightarrow\infty}\frac{1}{\overline{\gamma}}\ln\left(q_{n}(\jmath)\right)\leq-\overline{q}_{n}(\jmath).
\end{align}

We assume that, for $\jmath\neq 2^{N}-1$,
$
\overline{q}_{n}(\jmath)>0,~\forall n.
$
Since $\{n\}\subset\mathcal{I}^{n}_{k}$ for all $n$, necessarily for all $\jmath$, such that $n\notin \{i_{1},\cdots,i_{|\mathfrak{P}_\jmath|}\}$,
$\overline{q}_{n}(\jmath)=\infty$.
\end{itemize}
\begin{remark}
\label{rem-ass-obs} We now comment on the assumptions and justify their applicability under reasonable conditions (examples of distributed observation dissemination protocols with rate constraints are provided in Section~\ref{ex-dist}):
\begin{itemize}
\item[(i)] Assumption~\textbf{(E.1)} essentially means that the waiting times between successive transmissions are i.i.d. exponential random variables with mean $1/\overline{\gamma}$. This is justified in Carrier Sense Multiple Access (CSMA) type protocols, where the back-off time is often chosen to be exponentially distributed. To be more realistic, one needs to account for packet delays and transmission/reception processing times. We ignore these in the current setting. On a more practical note, the rate $\overline{\gamma}$ may be viewed as a function of the network communication bandwidth; the larger the bandwidth, the higher the rate of channel usages and hence $\overline{\gamma}$. In distributed network communication settings, a typical example of exponential waiting between successive transmissions is the asynchronous gossip model (see~\cite{Boyd-GossipInfTheory}).
\item[(ii)] Assumption~\textbf{(E.2)} is justified for memoryless and time-invariant communication schemes. It says that the scope of an instantiation of the distributed observation dissemination protocol is confined to the interval $[k\Delta,(k+1)\Delta)$, at the end of which the protocol restarts with a new set of observations independent of its past status. Equation \eqref{ass-obs1} is then a direct consequence of the Strong Law of Large Numbers (SLLN). This essentially means that the observation dissemination rate is $\overline{\gamma}$ times the observation acquisition or sampling rate scale.
\item[(iii)] Assumption~\textbf{(E.3)} is satisfied by any reasonable distributed protocol if the network is connected. Intuitively, this is due to the fact that, if $\overline{\gamma}>0$, the probability of having a sufficiently large number of communications in an interval of length $\Delta>0$ is strictly greater than zero (which can be very small though, depending on the value of $\overline{\gamma}$). On the other hand, if the network is connected, it is possible by using a sufficiently large (but finite) number of communications to disseminate the observation of a sensor to every other sensor. Examples of protocols satisfying~\textbf{(E.3)} are provided in Section~\ref{ex-dist}.
\item[(iv)] Assumption~\textbf{(E.4)} is justifiable by showing that ${q}_n(\jmath)$ decays exponentially as $\overline{\gamma} \rightarrow \infty$. Examples of protocols satisfying~\textbf{(E.4)} are provided in the next Section~\ref{ex-dist}.

\end{itemize}
\end{remark}
\begin{remark}
We claim that if random link failures are further considered in the protocol, the M-GIKF algorithm and the corresponding convergence result could still hold with minimum modification, since link failures basically lead to no information swapping or propagation between some particular node pairs. This results in the same effect as the case where the sensors choose to communicate with themselves in our current protocol. Apparently, with random link failures, to achieve the same error performance, it would require more signal evolution epochs compared with the case without link failures.
\end{remark}
\subsection{An Example of Distributed Observation Dissemination Protocol}
\label{ex-dist}
In this section, we give an example of a gossip based distributed observation dissemination protocol. During the epoch $[k\Delta,(k+1)\Delta)$, the protocol initiates the observation dissemination at sensor $n$. Sensor $n$ starts with its own current observation $\mathbf{y}^{n}_{k}$ and keeps exchanging its observation with its neighbors till the end of this epoch. The number of exchanges and the type of each exchange are determined by an asynchronous pairwise gossip protocol \cite{Boyd-GossipInfTheory}, where the inter-sensor communication occurs at successive ticks of a Poisson process with rate ${\overline{\gamma}}/{\Delta_o}$, and at each tick only one of the network links is active with uniform probability ${1}/{\overline{M}}$, where $\Delta_o$ is the time duration allocated for observation dissemination with each epoch, and $\overline{M}$ is the cardinality of the allowable communication link set $\mathcal{E}$. Equivalently, we could consider each network link activated independently of the others according to the ticks of a local Poisson clock with rate ${\overline{\gamma}}/{\Delta \overline{M}}$, where no two links will become active simultaneously due to the independence of events in the local Poisson processes. As a formal statement, the number of inter-sensor communications for observation dissemination $\mathcal{M}(k)$ in the interval $[k\Delta,(k+1)\Delta)$ follows a Poisson distribution with mean value $\overline{\gamma}$, which proves that this protocol satisfies assumption~\textbf{(E.1)}. In addition, the corresponding sequence of time-varying adjacency matrices $\{A^o_k(i)\}_{i=1,\cdots, \mathcal{M}(k)}$ is an i.i.d. sequence uniformly distributed on the set $\{E^{nl}\}$, where $E^{nl}$ is defined as a permutation matrix, such that, for each $(n,l)\in \mathcal{E}$ and $n\neq l$, $E^{nl}_{n,l}=E^{nl}_{l,n}=1$ and $E^{nl}_{m,m}=1$ for $m\neq n,l$, with all other entries being 0.

Now, we establish the observation dissemination process. Let $\mathbf{s}_k^i=[{s}_k^i(1),\cdots,{s}_k^i(N)]$ with its entry ${s}_k^i(n) \in [1,\cdots,N]$ indexing the observation $\mathbf{y}^{{s}_k^i(n)}_k$ at sensor $n$ just after the $i$-th exchange in the epoch $[k\Delta,(k+1)\Delta)$. Starting with ${s}_k^0(n)=n$ for each $n$ means that, at the beginning of the epoch $[k\Delta,(k+1)\Delta)$ before any exchanges, each sensor $n$ only has its own observation $\mathbf{y}^n_k$. When exchanges happen, the observations $\{\mathbf{y}^n_k\}_{1 \leq n \leq N}$ travel across the network according to
\begin{equation}
\mathbf{s}_k^i=A_k^o(i)\mathbf{s}_k^{i-1},~~i\in[1,\cdots,\mathcal{M}(k)].
\end{equation}

During this exchange process until the end of the epoch $[k\Delta,(k+1)\Delta)$, the sensors store the observations passing through them. Therefore, at the end of the epoch $[k\Delta,(k+1)\Delta)$, the set of observations available at sensor $n$ is
\begin{equation}
\mathcal{I}^{n}_k= \bigcup_{i=0}^{\mathcal{M}(k)}\{ {s}_k^{i} (n) \}.
\end{equation}

Finally, the observation dissemination for the epoch $[k\Delta,(k+1)\Delta)$ terminates at the end of this epoch, right before the sensor starts the next epoch $[(k+1)\Delta,(k+2)\Delta)$. Then similarly the observation dissemination repeats during the epoch $[(k+1)\Delta,(k+2)\Delta)$ independent of its prior state. Therefore, the sequence $\{\mathcal{I}^{n}_k\}$ as the set of observation indices available at sensor $n$ at the end of each epoch is a temporally i.i.d. process, which satisfies assumption~\textbf{(E.2)}. Moreover, this observation dissemination process is assumed to be independent of the estimate exchange process.

\begin{remark}
\label{rem:obs}
It is readily seen that the above observation dissemination protocol conforms to the preassigned gossip network structure. In fact, to execute the above protocol, each sensor needs to know its local communication neighborhood only, no global topology information is required. Secondly, note that, at each communication, a sensor forwards a single observation $\mathbf{y}_{k}^{s_{k}^{i}(n)}$ to a neighboring sensor. Since, the sensor observations are typically low-dimensional, the data overhead of each communication is modest. Finally, since the above protocol is fully randomized (neighbors are chosen independently uniformly), it is likely that a sensor will receive multiple copies of the same observation (possibly through different neighbors), i.e., some communications might end up being redundant.
\end{remark}

To prove that this protocol satisfies assumptions~\textbf{(E.3)} and~\textbf{(E.4)}, we have the following analysis employing the hitting time concept of Markov chains. For each $\jmath \neq 2^N-1$, without loss of generality, we assume that $\jmath$ corresponds to the sensor
subset $\{n_1,n_2,...,n_m\}$, with $\{n'_1,n'_2,...,n'_{N-m}\}$ denoting the complementary subset. As explained by the interacting particle representation in the next section, the link formation process following the sequence $\{A^o_k(i)\}$ for the observation dissemination can be represented as $N$ particles moving on the graph as identical Markov chains.
We use $T_{i}$ to denote the hitting time starting from sensor $i$ to another sensor $n$ in the Markov chain, with the
transition probability matrix as the mean adjacency matrix $\overline{A^o}$, which is irreducible and defined in a similar way as \eqref{def_barA}.
Then, we have
\begin{align}\label{upperbound_Oprotocol}
q_{n}(\jmath)&=P\left(T_{n'_1} >\mathcal{M}(k),\cdots,T_{n'_{N-m}} >\mathcal{M}(k),\right. \nonumber\\
&\left.~~~~~~~~~T_{n_1} \leq \mathcal{M}(k),\cdots,T_{n_m}\leq \mathcal{M}(k)\right)\nonumber\\
&\leq P\left(T_{n'_1} >\mathcal{M}(k),\cdots,T_{n'_{N-m}} >\mathcal{M}(k)\right) \nonumber\\
&\leq \min_{1 \leq i \leq {N-m}}P\left(T_{n'_i}>\mathcal{M}(k)\right).
\end{align}

From Theorem 7.26 in \cite{Driver-Markov}, since the transition matrix $\overline{A^o}$ is irreducible, there exist constants $ 0< \alpha <1$ and $ 0 <L < \infty$ such that $P(T_{i}>L) \leq \alpha,\forall i$, and more generally,
\begin{equation} \label{Driver1}
  P(T_{i}>kL) \leq \alpha^k, ~k=0,1,2,\cdots.
\end{equation}
Also, there exists constant $0<\beta <1$ such that $P(T_{i}>L) \geq \beta,\forall i$, and more generally,
\begin{equation}\label{Driver2}
  P(T_{i}>kL) \geq \beta^k, ~k=0,1,2,\cdots.
\end{equation}

Then, following \eqref{upperbound_Oprotocol}, we have
\begin{align}\label{upperbound_Oprotocol1}
\limsup_{\overline{\gamma} \rightarrow \infty} &\frac{1}{\overline{\gamma}} \ln\left(q_n(\jmath)\right) \nonumber \\
&\leq \limsup_{\overline{\gamma} \rightarrow \infty} \frac{1}{\overline{\gamma}} \ln \left(  \min_{1 \leq i \leq {N-m}}
P(T_{n'_i}>\mathcal{M}(k)) \right) \nonumber \\
&\leq  \limsup_{\overline{\gamma} \rightarrow \infty} \frac{1}{\overline{\gamma}} \ln \left( \alpha^{ \lfloor \frac{\mathcal{M}(k)}{L} \rfloor} \right)= \frac{\ln \, \alpha}{L}
\end{align}
where the last equation is obtained since $\lim_{\overline{\gamma} \rightarrow \infty}
\frac{\mathcal{M}(k)}{\overline{\gamma}}=1$.

We also have
\begin{align} \label{lowerbound_Oprotocol}
q_{n}(\jmath)&=P\left(T_{n'_1} > \mathcal{M}(k) ,\cdots,T_{n'_{N-m}} >\mathcal{M}(k) ,\right.\nonumber\\
&~~~~~~~~\left.T_{n_1} \leq \mathcal{M}(k),\cdots,T_{n_m} \leq \mathcal{M}(k)\right) \nonumber\\
&\geq P\left(T_{n'_1}>\mathcal{M}(k)\right)\cdots P\left(T_{n'_{N-m}} >\mathcal{M}(k)\right) \nonumber\\
&~~~P\left(T_{n_1} \leq \mathcal{M}(k) \right)\cdots P\left(T_{n_m} \leq \mathcal{M}(k) \right).
\end{align}

Then, from \eqref{Driver1} and \eqref{Driver2}, we have
\begin{align}\label{lowerbound_Oprotocol1}
\liminf_{\overline{\gamma} \rightarrow \infty}& \frac{1}{\overline{\gamma}} \ln\left(q_n(\jmath)\right) \nonumber\\
& \geq \liminf_{\overline{\gamma} \rightarrow \infty} \frac{1}{\overline{\gamma}} \ln\left[  \left(\beta^{ \lceil \frac{ \mathcal{M}(k)}{L} \rceil }\right)^{N-m}   \left(1- \alpha^{ \lfloor \frac{ \mathcal{M}(k)}{L} \rfloor }  \right)^m \right] \nonumber\\
& = (N-m) \frac{\ln \beta}{L}
\end{align}
where the last equation is obtained since $\lim_{\overline{\gamma} \rightarrow \infty}
\frac{\mathcal{M}(k)}{\overline{\gamma}}=1$ and $0 < \alpha <1$.

Therefore, from (\ref{upperbound_Oprotocol1}) and (\ref{lowerbound_Oprotocol1}), we have $\underline{q}_n(\jmath)$ and $\overline{q}_n(\jmath)$ in (\ref{E4_protocol}) well defined as
\begin{equation}\label{alpbeta}
   \underline{q}_n(\jmath)=(m-N) \frac{\ln \beta}{L} ,~~\overline{q}_n(\jmath)=- \frac{\ln\,\alpha}{L}.
\end{equation}

Since $\overline{q}_n(\jmath)= - \frac{\ln \alpha}{L}$ and $\alpha < 1$, clearly we see that, for $\jmath\neq 2^{N}-1$,
$\overline{q}_{n}(\jmath)>0$. 
Therefore, we have completed the proof that assumption~\textbf{(E.4)} holds.

To establish assumption~\textbf{(E.3)}, we denote $T_m=\max\{T_{1},...,T_{N}\}$, i.e., $T_m$ is the longest time among all hitting times to sensor $n$ from other sensors. Then, $q_{n}(2^{N}-1)=P( T_m \leq \mathcal{M}(k))=1-P(T_m > \mathcal{M}(k))$, which is greater than zero according to \eqref{Driver1}. This access to all the observations at the end of an epoch may be arbitrarily small but strictly greater than zero.

\subsection{An Interacting Particle Representation}
\label{int_part_rep}
As we did in \cite{Riccati-weakcons}, we now introduce an interesting particle representation, which bears subtle and critical differences from that in \cite{Riccati-weakcons} due to the extra observation dissemination procedure .

Recall from \cite{Riccati-weakcons} the notation $\mathbb{S}_{+}^{M}$ for the cone of positive semidefinite matrices and from Section~\ref{M-GIKF} the notations $\mathcal{C}_{\jmath}$ and $\mathcal{R}_{\jmath}$. To simplify the notation in (\ref{est_up2}), we define the functions of $f_{\jmath}:\mathbb{S}_{+}^{M}\longmapsto\mathbb{S}_{+}^{M}$ for $\jmath\in[0,\cdots, 2^{N}-1]$ denoting the respective subset Riccati operators\footnote{For $\jmath=0$,  the corresponding Riccati operator $f_{0}$ in (\ref{def_Riccati}) reduces to the Lyapunov operator, see \cite{Riccati-weakcons}.}:
\begin{equation}
\label{def_Riccati}
f_{\jmath}(X)=\mathcal{F}X\mathcal{F}^{T}+\mathcal{Q}-
\mathcal{F}X\mathcal{C}_{\jmath}^{T}\left(\mathcal{C}_{\jmath}X\mathcal{C}_{\jmath}^{T}+
\mathcal{R}_{\jmath}\right)^{-1}\mathcal{C}_{\jmath}X\mathcal{F}^{T}.
\end{equation}
Recall the sequence $n_k^{\rightarrow}$ of neighbors of sensor $n$. The sequence of conditional error covariance matrices $\left\{P^{n}_{k}\right\}$ at sensor $n$ then evolves according to
\begin{equation}
\label{evolve_P}\widehat{P}_{k+1}^{n}=f_{\jmath\,\scriptscriptstyle {(\mathcal{I}^{n}_{k})}}\left(\widehat{P}_{k}^{n_k^{\rightarrow}}\right)
\end{equation}
where ${\jmath\,{(\mathcal{I}^{n}_{k})}}$ denotes the index of $\mathcal{I}^{n}_{k}$ in the set $\mathfrak{P}$. The above sequence $\left\{\widehat{P}^{n}_{k}\right\}$ is non-Markovian (and is not even semi-Markov given the random adjacency matrix sequence $\{A(k)\}$), as $\widehat{P}_{k}^{n}$ at time $k$ is a random functional of the conditional error covariance of sensor $n_k^{\rightarrow}$ at time $k-1$, which, in general, is different from that of sensor $n$. This makes the evolution of the sequence $\left\{\widehat{P}^{n}_{k}\right\}$ difficult to track. To overcome this, we give the following interacting particle interpretation of the conditional error covariance evolution, from which we can completely characterize the evolution of the desired covariance sequences $\left\{\widehat{P}^{n}_{k}\right\}$ for $n=1,\cdots,N$.

To this end, we note that the link formation process given by the sequence $\{A(k)\}$ can be represented by $N$ particles moving on the graph as identical Markov chains. The state of the $n$-th particle is denoted by $z_{n}(k)$, and the sequence $\left\{z_{n}(k)\right\}_{k\in\mathbb{N}}$ takes values in $[1,\cdots,N]$. The evolution of the $n$-th particle is given as follows:
\begin{equation}
\label{part_n}z_{n}(k)=z_{n}(k-1)_k^{\rightarrow},~z_{n}(0)=n.
\end{equation}
Recall the (random) neighborhood selection $n_k^{\rightarrow}$. Thus, the $n$-th particle can be viewed as originating from node $n$ at time 0 and then traveling on the graph (possibly changing its location at each time) according to the link formation process $\{A(k)\}$. The following proposition establishes important statistical properties of the sequence $\left\{z_{n}(k)\right\}$ :
\begin{proposition}
\label{prop_int}
\item \mbox{(i)} For each $n$, the process $\left\{z_{n}(k)\right\}$ is a Markov chain on $V=[1,\cdots,N]$ with the transition probability matrix~$\overline{A}$.
\item \mbox{(ii)} The Markov chain $\left\{z_{n}(k)\right\}$ is ergodic with the uniform distribution on $V$ being the attracting invariant measure.
\end{proposition}
%

For each of the Markov chains $\left\{z_{n}(k)\right\}$, we define a sequence of switched Riccati iterates $\left\{P_{n}(k)\right\}$:
\begin{equation}
\label{swP}
P_{n}(k+1)=f_{\jmath \, {(\scriptstyle{\mathcal{I}_{z_{n}(k)}^{k}})}}(P_{n}(k)).
\end{equation}
The sequence $\left\{P_{n}(k)\right\}$ can be viewed as an iterated system of Riccati maps, in which the random switching sequence is governed by the Markov chain $\left\{z_{n}(k)\right\}$. A more intuitive explanation comes from the particle interpretation; precisely the $n$-th sequence may be viewed as a particle originating at node $n$ and hopping around the network as a Markov chain with transition probability $\overline{A}$ whose instantaneous state $P_{n}(k)$ evolves via the Riccati operator at its current location. In particular, in contrast to the sequence $\left\{\widehat{P}_k^{n}\right\}$ of the original conditional error covariances at sensor $n$, the sequence $\left\{P_{n}(k)\right\}$ does not correspond to the evolution of the error covariance at a particular sensor. The following proposition establishes the relation between $\left\{P_{n}(k)\right\}$ and the sequence $\left\{\widehat{P}_k^{n}\right\}$ of interest.
\begin{proposition}
\label{prop_semi}
Consider the sequence of random permutations $\left\{\pi_{k}\right\}$ on $V$, given by
\begin{equation}\label{swP100}
\left(\pi_{k+1}(1),\!\cdots\!,\pi_{k+1}(N)\right)\!=\!\left({\pi_{k}(1)}_k^{\rightarrow},\!\cdots\!,{\pi_{k}(N)}_k^{\rightarrow}\right)
\end{equation}
with initial condition
\begin{equation}
\label{swP101}
\left(\pi_{0}(1),\cdots,\pi_{0}(N)\right)=\left(1,\cdots,N\right).
\end{equation}
Note that $\pi_{k}(n)=z_{n}(k)$ for every $n$, where $z_{n}(k)$ is defined in~\eqref{part_n}. Then, for $k\in\mathbb{N}$,
\begin{equation}
\label{swP2}
\left(P_{1}(k),\cdots,P_{N}(k)\right)=\left(\widehat{P}_{k}^{\pi_{k}(1)},\cdots,\widehat{P}_{k}^{\pi_{k}(N)}\right).
\end{equation}
\end{proposition}
The above proposition suggests that the asymptotics of the desired sequence $\left\{\widehat{P}_k^{n}\right\}$ for every $n$ can be
obtained by studying the asymptotics for the sequences
$\left\{P_{n}(k)\right\}$. 
Hence, in the subsequent sections, we will focus on $\left\{P_{n}(k)\right\}$, rather than working
directly with the sequences $\left\{\widehat{P}_k^{n}\right\}$ of interest, which
involve a much more complicated statistical dependence.
\subsection{The Auxiliary Sequence $\left\{ \widetilde{P}(k)\right\}$}
\label{aux-RDS}
Since the switching Markov chains $\{z_n(k)\}$ are non-stationary, in order to analyze the processes $\left\{P_{n}(k)\right\}$ for $n=1,...,N$ under the scope of iterated random systems \cite{DiaconisFreedman} or RDSs \cite{Arnold}, we propose an auxiliary process $\left\{ \widetilde{P}(k)\right\}$ evolving with similar random Riccati iterates, but for which the corresponding switching Markov chain $\{\widetilde{z}(k)\}$ is stationary, i.e., $\{\widetilde{z}(k)\}$ is initialized by the uniform invariant measure on $V$. Then, we can analyze the asymptotic properties of the auxiliary sequence $\left\{ \widetilde{P}(k)\right\}$ by formulating it as an RDS on the space $\mathbb{S}_+^N$ and derive the asymptotics of the sequence $\left\{P_{n}(k)\right\}$ for $n=1,...,N$. The auxiliary sequence $\left\{ \widetilde{P}(k)\right\}$ is formally defined as follows, which follows the concept proposed in \cite{Riccati-weakcons}, but with necessary and non-trivial modifications to take into account observation dissemination.

Consider a Markov chain $\{\widetilde{z}(k)\}_{k\in \mathbb{T}_+}$ on the graph $V$, with transition matrix $\overline{A}$ and uniform initial distribution as follows:
\begin{equation}
  \mathbb{P} [ \widetilde{z}(0)=n]=\frac{1}{N},~~n=1,...,N.
\end{equation}
By proposition \ref{prop_int}, the Markov chain $\{\widetilde{z}(k)\}$ is stationary.

Now we can define the auxiliary process $\left\{ \widetilde{P}(k)\right\}$ with similar random Riccati iterates as
\begin{equation}
  \widetilde{P}(k+1)= f_{\jmath \, (\scriptstyle \mathcal{I}_{\widetilde{z}(k)}^{k})} \left(\widetilde{P}(k)\right)
\end{equation}
with (possibly random) initial condition $\widetilde{P}(0)$\footnote{Note that the sequences $\left\{P_{n}(k)\right\}$ of interest have deterministic initial conditions, but it is required for technical reasons to allow random initial states $\widetilde{P}(0)$ to study the auxiliary sequence $\left\{ \widetilde{P}(k)\right\}$.}.

In order to proceed with the asymptotic analysis of the auxiliary sequence $\left\{\widetilde{P}(k)\right\}$, we construct an RDS $(\theta,\varphi)$ on $\mathbb{S}_+^N$, equivalent to the auxiliary sequence $\left\{\widetilde{P}(k)\right\}$ in the sense of distribution. The construction process is similar to that in our previous paper \cite{Riccati-weakcons}; thus the details are omitted here. Briefly, denote $(\Omega, \mathcal{F}, \mathbb{P},\{\theta_k,k\in \mathbb{T}\})$ as a metric dynamical system, where $(\Omega, \mathcal{F}, \mathbb{P})$ is a probability space and the family of transformations $\{\theta_k\}_{k\in \mathbb{T}}$ on $\Omega$ is the family of left-shifts, i.e., $\theta_k w(\cdot)=w(k+\cdot),~\forall k\in \mathbb{T},~w\in \Omega $ ; the cocycle $\varphi:~\mathbb{T}_+\times\Omega \times \mathbb{S}_+^N \mapsto  \mathbb{S}_+^N$ is defined by
\begin{align}
\varphi(0,w,X)&=X,~\forall w,X, \nonumber\\
\varphi(1,w,X)&=f_{\jmath \,(\mathcal{I}_{w(0)}^{0})}(X),~\forall w,X, \nonumber\\
\varphi(k,w,X)&=f_{\jmath \,(\mathcal{I}_{\theta_{k-1}w(0)}^{k-1})}(\varphi(k-1,w,X))\nonumber\\
&=f_{\jmath \,(\mathcal{I}_{w(k-1)}^{k-1})}(\varphi(k-1,w,X)),~\forall k>1, w, X. \nonumber\label{cocycle3}
\end{align}

From the construction of $\{\theta,\varphi\}$, the sequence $\left\{ \varphi (k,w,{P}_{n}(0))\right\}_{k\in \mathbb{T}_+} $ is distributionally equivalent to the sequence $\left\{\widetilde{P}(k)\right\}_{k \in \mathbb{T}_+}$, i.e.,
\begin{equation}
\varphi(k,w,{P}_{n}(0))\overset{d}{=} \widetilde{P}(k),~\forall k\in\mathbb{T}_+.
\end{equation}
Therefore, analyzing the asymptotic distribution properties of the sequence $\left\{ \widetilde{P}(k)\right\}$ equals to studying the sequence $\left\{ \varphi (k,w,{P}_{n}(0))\right\}$, which we will analyze in the sequel.

We first establish some properties of the RDS $(\theta,\varphi)$ that represents the sequence $\left\{ \widetilde{P}(k)\right\}$.
\begin{lemma}\label{lem_RDSprop}
\item \mbox{(i)} The RDS $(\theta,\varphi)$ is conditionally compact.
\item \mbox{(ii)} The RDS $(\theta,\varphi)$ is order preserving.
\item \mbox{(iii)} If in addition $\mathcal{Q}$ is positive definite, i.e., $\mathcal{Q}\gg 0$, the RDS $(\theta,\varphi)$ is strongly sublinear.
\end{lemma}

The proof of Lemma \ref{lem_RDSprop} and the concepts including conditionally compact, order preserving, and sublinearity, are discussed in our prior work \cite{Riccati-weakcons}.
%
%

\section{Main Results and Discussions}
\label{main_res}
In this section, we present the main results of the paper, which include two parts. The first part concerns the asymptotic properties of the conditional error covariance processes at the sensors for a fixed $\overline{\gamma}>0$. These results generalize those obtained for the GIKF in~\cite{Riccati-weakcons}. The second part is the LD results concerning the family $\{\mathbb{\mu}\}$ of the $\{\mathbb{\mu}^{\overline{\gamma}}\}$ of the invariant filtering measures as $\overline{\gamma}\rightarrow\infty$, which is our key focus in this paper. As such, we just present the main results, the proofs of which are similar to those in \cite{Riccati-weakcons}.

\subsection{Asymptotic Results for Finite $\overline{\gamma}$}
\label{asy-aux}
We fix a $\overline{\gamma}>0$. First, we present the asymptotic properties of the auxiliary sequences $\left\{\widetilde{P}(k)\right\}$.
\begin{theorem} \label{asym_thm}
Under the assumptions \textbf{C.1}, \textbf{S.1}, and \textbf{D.1}, there exists a unique invariant probability measure $\mu^{\overline{\gamma}}$ on the space of positive semidefinite matrices $\mathbb{S}_+^N$, such that the sequence $\left\{ \widetilde{P}(k) \right\}$ converges weakly (in distribution) to $\mu^{\overline{\gamma}}$ from every initial condition $P_{n}(0)$ for each $n\in [1,\cdots,N]$, i.e.,
\begin{equation}
\left\{ \widetilde{P}(k) \right\} \Rightarrow \mu^{\overline{\gamma}}.
\end{equation}
\end{theorem}

Theorem~\ref{asym_thm} implies that the sequence $\left\{ \widetilde{P}(k) \right\}$ reaches consensus in the weak sense to the same invariant measure $\mu^{\overline{\gamma}}$ irrespective of the initial states, since $\mu^{\overline{\gamma}}$ does not depend on the index $n$ and on the initial state $\widetilde{P}(0)$ of the sequence $\left\{ \widetilde{P}(k) \right\}$.

Based on Theorem~\ref{asym_thm}, we can deduce Theorem \ref{main1}, which does not directly touch the sequences $\left\{\widehat{P}^{n}_k\right\}$ for $n=1,\cdots,N$, but sets the stage for showing the key result regarding the convergence of these sequences.
\begin{theorem}
\label{main1}
As defined in Section~\ref{aux-RDS}, $\left\{\widetilde{z}(k)\right\}$ is a stationary Markov chain on $V$ with transition probability matrix $\overline{A}$, i.e., $\widetilde{z}(0)$ is distributed uniformly on $V$. Let $\mathbb{\nu}$ be a probability measure on $\mathbb{S}_{+}^{M}$; and the process $\left\{\widetilde{P}(k)\right\}$ is given by
    \begin{equation}
    \label{main1:2}\widetilde{P}(k+1)=f_{\jmath \,(\mathcal{I}_{\widetilde{z}(k)}^{k})}\left(\widetilde{P}(k)\right),~k\in\mathbb{T}_{+}
    \end{equation}
    where $\widetilde{P}(0)$ is distributed as $\mathbb{\nu}$, independent of the Markov chain $\left\{\widetilde{z}(k)\right\}$ and the processes $\{\mathcal{I}_{n}^{k}\}$ for all $n$. Then,
    there exists a unique probability measure $\mathbb{\mu}^{\overline{\gamma}}$ such that, for every $\mathbb{\nu}$, the process $\left\{\widetilde{P}(k)\right\}$ constructed above converges weakly to $\mathbb{\mu}^{\overline{\gamma}}$ as $k\rightarrow\infty$, i.e.,
    \begin{equation}
    \label{main1:1} f_{\jmath \,(\mathcal{I}_{\widetilde{z}(k)}^{k})}\circ f_{\jmath \,(\mathcal{I}_{\widetilde{z}(k-1)}^{k-1})}\cdots\circ f_{\jmath \,(\mathcal{I}_{\widetilde{z}(0)}^{0})}\left(\widetilde{P}(0)\right)\Longrightarrow\mu^{\overline{\gamma}}.
\end{equation}
\end{theorem}
We now state the theorem characterizing the convergence properties of the sequences $\left\{\widehat{P}_k^{n}\right\}$.
\begin{theorem}\label{main2}
Let $q$ be a uniformly distributed random
variable on $V$, independent of the sequence of adjacency
matrices $\{A(k)\}$ and the processes $\{\mathcal{I}_{n}^{k}\}$. Then, the sequence
$\left\{\widehat{P}_k^{q}\right\}$ converges weakly to
$\mathbb{\mu}^{\overline{\gamma}}$ defined in
Theorem~\ref{main1}, i.e.,
\begin{equation}\label{main2:1}
    \widehat{P}_k^{q}\Longrightarrow\mu^{\overline{\gamma}}.
\end{equation}
In other words, the conditional error covariance $\left\{\widehat{P}_k^{q}\right\}$ of a randomly selected sensor converges in distribution to $\mathbb{\mu}^{\overline{\gamma}}$.

\end{theorem}
\begin{remark}
\label{rem-main2}
Theorem \ref{main2} reinforces the weak consensus achieved by the M-GIKF, i.e., the conditional error covariance at a randomly selected sensor converges in distribution to an invariant measure $\mu^{\overline{\gamma}}$. In other words, it provides an estimate $\{\widehat{\mathbf{x}}_q(k)\}$ for the entire signal $\mathbf{x}$, where $\{\widehat{\mathbf{x}}_q (k)\}$ could be obtained by uniformly selecting a sensor $q$ and using its estimate $\{\widehat{\mathbf{x}}_q (k)\}$ for all time $k$. Also, note that the results here pertain to the limiting distribution of the conditional error covariance and hence the pathwise filtering error. This is a much stronger result than just providing the moment estimates of the conditional error covariance, which does not provide much insight into the pathwise instantiation of the filter. In the following subsection, we provide an analytical characterizations of the invariant measure $\mu^{\overline{\gamma}}$ by showing that it satisfies the LD lower and upper bounds as $\mu^{\overline{\gamma}} \rightarrow \infty$ .
%
\end{remark}

\subsection{Large Deviation Probabilities for the $\{\mathbb{\mu}^{\overline{\gamma}}\}$ Family}
\label{ldp-mu}
In this section, we characterize the invariant measure $\mathbb{\mu}^{\overline{\gamma}}$ governing the asymptotics of the conditional sensor error covariance process $\{\widehat{P}_{k}^{n}\}$,~$n=1,\cdots,N$. The following result is a first step to understanding the behavior of the invariant distribution $\mathbb{\mu}^{\overline{\gamma}}$ family.

\begin{theorem}\label{thm_convDirac}
The family of invariant distributions $\mathbb{\mu}^{\overline{\gamma}}$ converges weakly, as $\overline{\gamma} \rightarrow \infty$, to the Dirac measure $\delta_{P^{\ast}}$ corresponding to the performance of the centralized estimator (recall, $P^{\ast}$ is the unique fixed point of the centralized Riccati operator $f_{2^{N}-1}$).
\end{theorem}

\begin{remark}
Theorem~\ref{thm_convDirac} states that the family $\{\mu^ {\overline{\gamma}} \}$ converges weakly to the Dirac measure $\delta_{P^{\ast}}$ concentrated at $P^*$, as ${\overline{\gamma}} \rightarrow \infty$, which is intuitive, since with ${\overline{\gamma}} \rightarrow \infty$, the distributed M-GIKF filtering process reduces to classical Kalman filtering with all the observations available at a fusion center, i.e., centralized filtering, where  $P^{\ast}$ is the unique fixed point of this centralized filtering. Therefore, with ${\overline{\gamma}} \rightarrow \infty$, we expect the M-GIKF to perform more and more similarly to the centralized case, which leads to the weak convergence of the measure $\mu^ {\overline{\gamma}}$ to $\delta_{P^{\ast}}$ as ${\overline{\gamma}} \rightarrow \infty$. An immediate consequence of Theorem~\ref{thm_convDirac} is
\begin{equation}
\lim_{\overline{\gamma} \rightarrow \infty} {\mu}^{\overline{\gamma}} (\Gamma) =0, ~~\forall \overline{\Gamma} \cap P^*= \emptyset
\end{equation}
which means, w.r.t. $\{\mu^{\overline{\gamma}}\}$, every event $\Gamma$ with $P^* \notin \overline{\Gamma}$ is a rare event. This is intuitively correct, since as $\overline{\gamma} \rightarrow \infty$, the measures $\{{\mu}^{\overline{\gamma}}\}$ become more and more concentrated on an arbitrarily small neighborhood of $P^*$, resulting in the event $\Gamma$ becoming very difficult to observe.
\end{remark}

The proof of this theorem is presented in Appendix~\ref{proof_convDiarac}. In the sequel, we establish the LD upper and lower bounds for the family $\{{\mu}^{\overline{\gamma}}\}$ as $\overline{\gamma} \rightarrow \infty$, which completely characterizes the behavior of $\{{\mu}^{\overline{\gamma}}\}$.

Recall the set of strings $\overline{\mathcal{S}}$ in Definition~\ref{string-def}. For an integer $r\geq 1$, let $\mathcal{P}_{r}$ denote the set of all paths of length $r$ in the sensor graph w.r.t. the adjacency matrix $\mathcal{A}$, i.e.,
\begin{eqnarray}
\label{ldp-mu1}
\mathcal{P}_{r} &=&\left\{(n_{r},\cdots,n_{1})~|~n_{i}\in [1,\cdots,N],~\forall 1\leq i\leq r~\mbox{and}~\right.\nonumber \\ & & \left.\mathcal{A}_{n_{i},n_{i+1}}>0,~\forall 1\leq i<r\right\}.
\end{eqnarray}
To each string $\mathcal{R}=(f_{\jmath_{r}},f_{\jmath_{r-1}},\cdots,f_{\jmath_{1}},P^{\ast})\in\mathcal{S}_{r}^{P^{\ast}}$ of length $r$, defined in Section~\ref{app_res}, we assign its upper and lower weights respectively as,
\begin{equation}
\label{up-weight}
\overline{w}(\mathcal{R})=\min_{(n_{r},\cdots,n_{1})\in\mathcal{P}_{r}}\sum_{i=1}^{r}\mathbb{I}_{\jmath_{i}\neq 2^{N}-1}\overline{q}_{n_{i}}(\jmath_{i})
\end{equation}
\begin{equation}
\label{up-weight1}
\underline{w}(\mathcal{R})=\min_{(n_{r},\cdots,n_{1})\in\mathcal{P}_{r}}\sum_{i=1}^{r}\mathbb{I}_{\jmath_{i}\neq 2^{N}-1}\underline{q}_{n_{i}}(\jmath_{i}).
\end{equation}
We set $\overline{w}(\mathcal{R})=\underline{w}(\mathcal{R})=0$, if $r=0$ in the above.

Note that $|\mathcal{P}_{r}|<\infty$ for each $r\in\mathbb{N}$; hence $\overline{w}(\cdot)$ and $\underline{w}(\cdot)$ are well-defined extended valued functions mapping from $\mathcal{S}^{P^{\ast}}$ to ${\mathbb{R}}_{+}$ (we adopt the convention that the minimum of an empty set is $\infty$).

Finally, define the upper and lower rate functions, $\overline{I},\underline{I}:\mathbb{S}_{+}^{M}\longmapsto{\mathbb{R}}_{+}$ by
\begin{equation}
\label{up-rate}
\underline{I}(X)=\inf_{\mathcal{R}\in\mathcal{S}^{P^{\ast}}(X)}\underline{w}(\mathcal{R}),~\overline{I}(X)=\inf_{\mathcal{R}\in\mathcal{S}^{P^{\ast}}(X)}\overline{w}(\mathcal{R}).
\end{equation}
We then have the following large deviation results for the family $\left\{\mathbb{\mu}^{\overline{\gamma}}\right\}$ as $\overline{\gamma}\rightarrow\infty$.
\begin{theorem}\label{th:ldp}
 Assume that~\textbf{(C.1)},~\textbf{(S.1)},~\textbf{(D.1)},~and \textbf{(E.4)} hold. Then, as $\overline{\gamma}\rightarrow\infty$, the family $\mathbb{\mu}^{\overline{\gamma}}$ satisfies the LD upper and lower bounds with rate functions $\overline{I}$ and $\underline{I}$, i.e.,
\begin{equation}
\label{th:ldp1}
\limsup_{\overline{\gamma}\rightarrow\infty}\frac{1}{\overline{\gamma}}\ln\mathbb{\mu}^{\overline{\gamma}}\!\left(\mathcal{F}\right)\!\leq\! -\inf_{X\in\mathcal{F}}\overline{I}(X),~\mbox{for every closed set $\mathcal{F}\in\mathcal{X}$}
\end{equation}
\begin{equation}
\label{th:ldp2}
\liminf_{\overline{\gamma}\rightarrow\infty}\frac{1}{\overline{\gamma}}\ln\mathbb{\mu}^{\overline{\gamma}}\!\left(\mathcal{O}\right)\!\geq\! -\inf_{X\in\mathcal{O}}\underline{I}(X),~\mbox{for every open set $\mathcal{O}\in\mathcal{X}$}.
\end{equation}
\end{theorem}

\begin{remark}
Theorem~\ref{th:ldp} characterizes the invariant measure $\{\mathbb{\mu}^{\overline{\gamma}}\}$ as $\overline{\gamma} \rightarrow \infty$. It establishes the important qualitative behavior of $\{\mathbb{\mu}^{\overline{\gamma}}\}$ that rare events decay exponentially when $\overline{\gamma} \rightarrow \infty$. For a rare event $\Gamma$, from (\ref{th:ldp1}) and (\ref{th:ldp2}), we have
\begin{equation}
e^{-\overline{\gamma}(\inf_{X\in\Gamma^{\circ}}\underline{I}(X))} \leq \mathbb{\mu}^{\overline{\gamma}}(\Gamma)\leq e^{-\overline{\gamma}(\inf_{X\in\overline{\Gamma}}\overline{I}(X))}.
\end{equation}

The exact exponent of the exponential decay is bounded within $[{\overline{\gamma}(\inf_{X\in\overline{\Gamma}}\overline{I}(X))},\overline{\gamma}(\inf_{X\in\Gamma^{\circ}}\underline{I}(X))]$. The result suggests how the system designer could trade off estimation accuracy with the communication rate $\overline{\gamma}$. For instance, given a tolerance $\varepsilon>0$, in order to guarantee the probability of estimation errors lying outside the $\varepsilon-$neighborhood of the optimal centralized estimation error $P^*$ is less than some $\delta > 0$, $\overline{\gamma}$ should be selected according to
\begin{equation}
  e^{-\overline{\gamma}\left(\inf_{X\in B_\varepsilon^C(P^*)^{\circ}}\underline{I}(X)\right)}\!\leq \mathbb{\mu}^{\overline{\gamma}}(B_\varepsilon^C(P^*)) \leq \! e^{-\overline{\gamma}\left(\inf_{X\in\overline{B_\varepsilon^C(P^*) }}\overline{I}(X)\right)}\nonumber
\end{equation}
where $B_\varepsilon^C(P^*)$ is the complement of the open ball $B_\varepsilon(P^*)$. By computing $\inf_{X\in B_\varepsilon^C(P^*)^{\circ}}\underline{I}(X)$ and $\inf_{X\in\overline{B_\varepsilon^C(P^*) }}\overline{I}(X)$, the designer obtains an estimate of the communication rate $\overline{\gamma}$ required to maintain the probability of outlying errors less than $\delta$.
\end{remark}
The rest of the paper is devoted to the proof of Theorem~\ref{th:ldp} in Sections~\ref{ld-inv} and \ref{proof_th:ldp}, and some intermediate results are presented in Section~\ref{app_res}.

\section{Some Intermediate Results}
\label{app_res}
In this section we present first results on the random compositions of Riccati operators, to be used for analyzing the switched sequences of form~\eqref{swP}, and then approximation results generalizing those in~\cite{Riccati-moddev}.

\subsection{Preliminary Results}
\label{notation}The RRE sequence is an iterated function system (see, e.g.,~\cite{DiaconisFreedman}) comprising of random compositions of Riccati operators. Understanding the system requires studying the behavior of such random function compositions, where not only the numerical value of the composition is important, but also the composition pattern is relevant. To formalize this study, we start with the following definitions.

\begin{definition}[String]
\label{string-def}
Let $P_{0}\in\mathbb{S}_{+}^{M}$. A string
$\mathcal{R}$ with initial state $P_{0}$ and length $r\in\mathbb{N}$ is a $(r+1)$-tuple of the form:
\begin{equation}
\label{def_string} \mathcal{R}=\left(f_{\jmath_{r}}, f_{\jmath_{r-1}},\cdots
f_{\jmath_{1}},P_{0}\right),~~~\jmath_{1},\cdots,\jmath_{r}\in\mathfrak{P}
\end{equation}
where $f_{\jmath}$ corresponds to the Riccati operator defined in~\eqref{def_Riccati}. The length of a string $\mathcal{R}$ is denoted by $\mbox{len}(\mathcal{R})$.
The set of all possible strings is denoted by
$\overline{\mathcal{S}}$.

Fix $\overline{\gamma}>0$. A string $\mathcal{R}$ of the form
\[
\mathcal{R}=\left(f_{\jmath_{r}}, f_{\jmath_{r-1}},\cdots
f_{\jmath_{1}},P_{0}\right),~~~\jmath_{1},\cdots,\jmath_{r}\in\mathfrak{P}
\]
is called $\overline{\gamma}$-feasible, if there exists a path\footnote{A sequence of nodes $(n_{r},n_{r-1},\cdots,n_{1})$ is called a path w.r.t. $\overline{A}$ if $\overline{A}_{n_{i},n_{i+1}}>0$ for all $1\leq i<r$.} $(n_{r},n_{r-1},\cdots,n_{1})$ of length $r$ w.r.t. $\overline{A}$, such that $q_{n_{i}}(\jmath_{i})>0$ (recall $q_{n_{i}}(\jmath_{i})$ defined in \eqref{ass-obs4}) for all $1\leq i\leq r$. The set of all $\overline{\gamma}$-feasible strings is further denoted by $\overline{\mathcal{S}}_{\overline{\gamma}}$.
\end{definition}

\begin{remark} Note that a string $\mathcal{R}$ can be of length 0; then it is represented as a 1-tuple, consisting of only the initial condition.

Let $r_{1},r_{2},\cdots,r_{l}$ be non-negative integers, such that $\sum_{i=1}^{l}r_{i}=r$ and $\jmath^{k}_{i}\in\mathfrak{P}$ for $1\leq i\leq r_{k}$ and $1\leq k\leq l$, where for all $k$, we have $\jmath^{k}_{i}=\jmath^{k}_{1},~1\leq i\leq r_{k}$. Let $\mathcal{R}$ be a string of length $r$ of the following form:
\begin{equation}
\label{string_def_remark}
\mathcal{R}\!=\!\left(f_{\jmath^{1}_{1}},\!\cdots\!,f_{\jmath^{1}_{r_{1}}},\!\cdots\!,f_{\jmath^{2}_{1}},\!\cdots\!,f_{\jmath^{2}_{r_{2}}},\!\cdots\!,f_{\jmath^{l}_{1}},\!\cdots\!,f_{\jmath^{l}_{r_{l}}},P_{0}\right).
\end{equation}
For brevity, we write $\mathcal{R}$ as
\begin{equation}
\label{string_def_remark1}
\mathcal{R}=\left(f_{\jmath^{1}_{1}}^{r_{1}},f_{\jmath^{2}_{1}}^{r_{2}},\cdots,f_{\jmath^{l}_{1}}^{r_{l}},P_{0}\right).
\end{equation}
For example, the string $\left(f_{1},f_{2},f_{2},f_{2},f_{1},f_{1},P_{0}\right)$ could be written concisely as $\left(f_{1},f_{2}^{3},f_{1}^{2},P_{0}\right)$.
\end{remark}

\begin{definition}[Numerical Value of a String] Every string $\mathcal{R}$ is associated with its numerical
value, denoted by $\mathcal{N}(\mathcal{R})$, which is the
numerical evaluation of the function composition on the initial
state $P_{0}$; i.e., for $\mathcal{R}$ of the form \[\mathcal{R}=\left(f_{\jmath_{r}}, f_{\jmath_{r-1}},\cdots
f_{\jmath_{1}},P_{0}\right),~~\jmath_{1},\cdots,\jmath_{r}\in\mathfrak{P},\] we have
\begin{equation}
\label{num_string100}
\mathcal{N}(\mathcal{R})=f_{\jmath_{r}}\circ f_{\jmath_{r-1}}\circ\cdots\circ f_{\jmath_{1}}(P_{0}).
\end{equation}
Thus\footnote{For function compositions, we adopt a similar notation to that of strings; for example, we denote the composition $f_{1}\circ f_{2}\circ f_{2}\circ f_{2}\circ f_{1}\circ f_{1}(P_{0})$ by $f_{1}\circ f_{2}^{3}\circ f_{1}^{2}(P_{0})$.}, the numerical value can be viewed as a
function $\mathcal{N}(\cdot)$ mapping from the space
$\overline{\mathcal{S}}$ of strings to $\mathbb{S}_{+}^{M}$. We
abuse notation by denoting $\mathcal{N}(\overline{\mathcal{S}})$
as the set of numerical values attainable, i.e.,
\begin{equation}
\label{def_numval1}
\mathcal{N}(\overline{\mathcal{S}})=\left\{\mathcal{N}(\mathcal{R})~|~\mathcal{R}\in\overline{\mathcal{S}}\right\}.
\end{equation}
Similarly, by $\mathcal{N}(\overline{\mathcal{S}}_{\overline{\gamma}})$ we denote the subset of numerical values associated to the $\overline{\gamma}$-feasible strings $\overline{\mathcal{S}}_{\overline{\gamma}}$.
\end{definition}

\begin{remark}Note the difference between a string and its
numerical value. Two strings are equal if and only if they comprise the same order of function compositions applied to the same initial state.
In particular, two strings can be different, even if they are evaluated with the same numerical value.
\end{remark}

For fixed $P_{0}\in\mathbb{S}_{+}^{M}$ and $r\in\mathbb{N}$, the subset of strings of length $r$ and initial condition $P_{0}$ is denoted by $\mathcal{S}_{r}^{P_{0}}$. The corresponding set of numerical values is denoted by $\mathcal{N}(\mathcal{S}_{r}^{P_{0}})$. Finally, for $X\in\mathbb{S}_{+}^{M}$, the set $\mathcal{S}_{r}^{P_{0}}(X)\subset\mathcal{S}_{r}^{P_{0}}$ consists of all strings with numerical value $X$, i.e.,
\begin{equation}
\label{concat_string3}
\mathcal{S}_{r}^{P_{0}}(X)=\left\{\mathcal{R}\in\mathcal{S}_{r}^{P_{0}}~|~\mathcal{N}\left(\mathcal{R}\right)=X\right\}.
\end{equation}

In the following, we present some important
properties of strings to be used later. Recall from \cite{Riccati-weakcons} that $\mathbb{S}_{++}^{M}$ is the cone of positive definite matrices.
\begin{proposition}
\label{string_prop}
\item \mbox{(i)} For $r_{1}\leq r_{2}\in\mathbb{N}$, we have $\mathcal{N}\left(\mathcal{S}_{r_{1}}^{P^{\ast}}\right)\subset\mathcal{N}\left(\mathcal{S}_{r_{2}}^{P^{\ast}}\right)$, where $P^{\ast}\in\mathbb{S}_{++}^{M}$ denotes the unique fixed point of the Riccati operator $f_{2^{N}-1}$.
In particular, if for some $X\in\mathbb{S}_{+}^{M}$, $r_{0}\in\mathbb{N}$, and $\jmath_{r_{0}},\cdots, \jmath_{1}\in\mathfrak{P}$, the string $\mathcal{R}=\left(f_{\jmath_{r_{0}}},\cdots,f_{\jmath_{1}},P^{\ast}\right)$ belongs to $\mathcal{S}_{r_{0}}^{P^{\ast}}(X)$, we have
\begin{equation}
\label{def_StP200}
\left(f_{\jmath_{r_{0}}},\cdots,f_{\jmath_{1}}, f_{2^{N}-1}^{r-r_{0}},P^{\ast}\right)\in\mathcal{S}_{r}^{P^{\ast}}(X)\subset\mathcal{S}^{P^{\ast}}(X),~\forall r\geq r_{0}.
\end{equation}

\item \mbox{(ii)} Let $r\in\mathbb{N}$ and
$\mathcal{R}\in\mathcal{S}_{r}^{P_{0}}=\left(f_{\jmath_{r}},\cdots,
f_{\jmath_{1}},P_{0}\right)$ be a string. Define the function $\pi(\cdot)$ by
\begin{equation}
\label{concatS5} \pi\left(\mathcal{R}\right)=\left\{
\begin{array}{ll}
                    \sum_{i=1}^{r}\left(1-\mathbb{I}_{\{2^{N}-1\}}(\jmath_{i})\right), & \mbox{if $r\geq 1$} \\
                    0, & \mbox{otherwise}.
                   \end{array}
          \right.
\end{equation}
i.e., $\pi(\mathcal{R})$ counts the number of occurrences of the non-centralized Riccati operator $f_{2^{N}-1}$ in $\mathcal{R}$.

Also denote $\widehat{\mathcal{R}}=(f_{\hat{\jmath}_{\pi(\mathcal{R})}},f_{\hat{\jmath}_{\pi(\mathcal{R})-1}},\cdots,f_{\hat{\jmath}_{1}},P_{0})$, which represents the string of length $\pi(\mathcal{R})$ obtained by removing the occurrences of $f_{2^{N}-1}$ from $\mathcal{R}$\footnote{For example, if $\mathcal{R}=(f_{1},f_{2^{N}-1},f_{3},f_{2^{N}-1},f_{2},P_{0})$, $\widehat{\mathcal{R}}=(f_{1},f_{3},f_{2},P_{0})$.}.

Then, there exists
$\alpha_{P_{0}}\in\mathbb{R}_{+}$, depending on $P_{0}$ only, such
that
\begin{equation}
\label{concatS4}
f_{\hat{\jmath}_{\pi(\mathcal{R})}}\circ f_{\hat{\jmath}_{\pi(\mathcal{R})-1}}\cdots\circ f_{\hat{\jmath}_{1}}\left(\alpha_{P_{0}}I\right)\succeq\mathcal{N}\left(\mathcal{R}\right).
\end{equation}
\end{proposition}
\begin{proof} The proof is a straightforward generalization of Proposition 3.6 in~\cite{Riccati-moddev} and is omitted.
\end{proof}

\subsection{Riccati Equation}
\label{key_approx} In this subsection, we present several
approximation results needed in the sequel.
We discuss generic properties, like uniform convergence of the
Riccati operator, which will be used in the sequel to obtain various
tightness estimates required for establishing the LD results.
\begin{proposition} \label{unif_conv}
\item \mbox{(i)} For every $X\in\mathbb{S}_{+}^{M}$ and $\jmath\in[0,\cdots, 2^{N}-1]$, we have
\begin{equation}
\label{unif_conv2000}
f_{\jmath}(X)\succeq f_{2^{N}-1}(X).
\end{equation}
\item \mbox{(ii)} For every $\varepsilon>0$, there exists
$r_{\varepsilon}\geq M$, such that, for every
$X\in\mathbb{S}_{+}^{M}$, with $X\succeq P^{\ast}$ ($P^{\ast}$ is the unique fixed point of the centralized Riccati operator $f_{2^{N}-1}$), we have
\begin{equation}
\label{unif_conv1}
\left\|f_{2^{N}-1}^{r}\left(X\right)-P^{\ast}\right\|\leq\varepsilon,~~~r\geq r_{\varepsilon}.
\end{equation}
Note, in particular, that $r_{\varepsilon}$ can be chosen
independent of the initial state $X$.
\item \mbox{(iii)} For a fixed $r\in\mathbb{N}$ and
$\jmath_{r},\cdots,\jmath_{1}\in\mathfrak{P}$, define the function
$g:\mathbb{S}_{+}^{M}\longmapsto\mathbb{S}_{+}^{M}$ by
\begin{equation}
\label{Lip1} g(X)=f_{\jmath_{r}}\circ\cdots\circ
f_{\jmath_{1}}(X),~~~X\in\mathbb{S}_{+}^{M}.
\end{equation}
Then $g(\cdot)$ is Lipschitz continuous with some constant
$K_{g}>0$. Also, for every $\varepsilon_{2}>0$, there exists
$r_{\varepsilon_{2}}$, such that the function
$f_{2^{N}-1}^{r_{\varepsilon_{2}}}(\cdot)$ is Lipschitz continuous with
constant $K_{f_{2^{N}-1}^{r_{\varepsilon_{2}}}}<\varepsilon_{2}$.
\end{proposition}
\begin{proof}
The second and third assertions follow from Lemmas 3.1 and 3.2 in~\cite{Riccati-moddev}.
For the the first assertion, note that by (\ref{def_Riccati}) we have
\begin{equation}
f_{\jmath}(X)=\mathcal{F}X\mathcal{F}^{T}+\mathcal{Q}-
\mathcal{F}X\mathcal{C}_{\jmath}^{T}\left(\mathcal{C}_{\jmath}X\mathcal{C}_{\jmath}^{T}+
\mathcal{R}_{\jmath}\right)^{-1}\mathcal{C}_{\jmath}X\mathcal{F}^{T}
\end{equation}
where $\mathcal{C}_{\jmath}=[\mathcal{C}_{i_{1}}^{T}\cdots\mathcal{C}_{i_{|\jmath|}}^{T}]^{T}$. Hence, we can rewrite the above equation as
\begin{equation}\label{eq_fjmax}
f_{\jmath}(X)\!=\!\mathcal{F}X\mathcal{F}^{T}\!+\mathcal{Q}-\!\sum_{j=1}^{|\jmath|}
\mathcal{F}X\mathcal{C}_{i_j}^{T}\left(\mathcal{C}_{i_j}X\mathcal{C}_{i_j}^{T}\!+\!
\mathcal{R}_{i_j}\right)^{-1}\!\mathcal{C}_{i_j}X\mathcal{F}^{T}
\end{equation}
that is obtained due to the fact that $\mathcal{C}_{\jmath}X\mathcal{C}_{\jmath}^{T}+
\mathcal{R}_{\jmath}$ is block diagonal, from which it follows that
\begin{align}
&\left(\mathcal{C}_{\jmath}X\mathcal{C}_{\jmath}^{T}+
\mathcal{R}_{\jmath}\right)^{-1}\nonumber\\
&=\begin{pmatrix}
     \left(\mathcal{C}_{i_1}X\mathcal{C}_{i_1}^{T}+
\mathcal{R}_{i_1}\right)^{-1} \\
     & & \ddots\\
     & &  & \left(\mathcal{C}_{i_{|\jmath|}}X\mathcal{C}_{i_{|\jmath|}}^{T}+
\mathcal{R}_{i_{|\jmath|}}\right)^{-1}
 \end{pmatrix}.
\end{align}

Since $\jmath=2^N-1$ corresponds to the entire set of nodes $\{1,\cdots,N\}$, we have $|\jmath|\leq N$ and hence
\begin{align}
&f_{2^N-1}(X)=\mathcal{F}X\mathcal{F}^{T}+\mathcal{Q}\nonumber\\
&-\sum_{j=1}^{N}\mathcal{F}X\mathcal{C}_{i_j}^{T}\left(\mathcal{C}_{i_j}X\mathcal{C}_{i_j}^{T}+
\mathcal{R}_{i_j}\right)^{-1}\mathcal{C}_{i_j}X\mathcal{F}^{T}.
\end{align}
Therefore, $f_{\jmath}(X)\succeq f_{2^{N}-1}(X)$.
\end{proof}

\section{LD for Invariant Measures}
\label{ld-inv}
In this section, we first define the lower semicontinuous regularization $\overline{I}_L$ of $\overline{I}$ and the lower semicontinuous regularization $\underline{I}_L$ of $\underline{I}$ and establish their properties in Subsection~\ref{prop-rate}. Then, we are set to establish the LD for the family of invariant measures $\left\{\mu^{\overline{\gamma}}\right\}$, with $\overline{I}_L$ and $\underline{I}_L$ working as the upper and lower rate functions, respectively.

\subsection{The Upper and Lower Rate Functions}
\label{prop-rate}
We define
\begin{equation}
{I}(X)=\inf_{\mathcal{R}\in\mathcal{S}^{P^{\ast}}(X)}\pi(\mathcal{R}),~\forall X \in \mathbb{S}_{+}^{M}.
\end{equation}
Recall that
$\overline{I},\underline{I}:\mathbb{S}_{+}^{M}\longmapsto\overline{\mathbb{R}}_{+}$ are defined as
\begin{align}
\label{up-rate1}
\underline{I}(X)\!=\!\inf_{\mathcal{R}\in\mathcal{S}^{P^{\ast}}(X)}\underline{w}(\mathcal{R}),\!~\overline{I}(X)\!=\!\inf_{\mathcal{R}\in\mathcal{S}^{P^{\ast}}(X)}\overline{w}(\mathcal{R}),\!~\forall X \!\in \mathbb{S}_{+}^{M}.
\end{align}
The functions $\overline{I},\underline{I}$ are not generally lower semicontinuous and hence do not qualify as rate functions. However, candidate rate functions for the family of invariant distributions can be the lower semicontinuous regularizations of $\overline{I},\underline{I}$, which are defined as
\begin{align}
&\overline{I}_L(X)= \lim_{\varepsilon \rightarrow \infty} \inf_{Y\in B_{\varepsilon}(X)} \overline{I}(Y),~\forall X \in \mathbb{S}_{+}^{M}\nonumber\\
&\underline{I}_L(X)= \lim_{\varepsilon \rightarrow \infty} \inf_{Y\in B_{\varepsilon}(X)} \underline{I}(Y), ~\forall X \in \mathbb{S}_{+}^{M}.
\end{align}
The following proposition gives some readily verifiable properties of $\overline{I}_L(X)$, whose proof may be obtained from Proposition 6.1 of \cite{Riccati-moddev}.The semicontinuous regularization $\underline{I}_L(X)$ also has similar properties.
\begin{proposition}\label{prop_rate_function}
\item \mbox{(i)} The function $\overline{I}_L(X)$ is a good rate function on $\mathbb{S}_{+}^{M}$.
\item \mbox{(ii)} For every $X\in \mathbb{S}_{+}^{M} $, $\overline{I}_L(X)= \lim_{\varepsilon \rightarrow 0} \inf_{Y\in \overline{B_{\varepsilon}(X)}} \overline{I}(Y) $.
\item \mbox{(iii)} For every non-empty set $\Gamma \in \mathcal{B} (\mathbb{S}_{+}^{M})$, $ \inf_{X \in \Gamma} \overline{I}_L(X) \leq \inf_{X \in \Gamma} \overline{I}(X) $. In addition, if $\Gamma$ is open, the reverse inequality holds and thus $\inf_{X \in \Gamma} \overline{I}_L(X) = \inf_{X \in \Gamma} \overline{I}(X) $.
\item \mbox{(iv)} Let $K \subset \mathbb{S}_{+}^{M}$ be a non-empty compact set; then we have $\lim_{\varepsilon \rightarrow 0} \inf_{Y \in \overline{K_{\varepsilon}}} \overline{I}_L(Y)= \inf_{Y \in K} \overline{I}_L(Y) $.
\end{proposition}

\subsection{The LD lower bound}
\label{ld-lower}
The following lemma establishes the LD lower bound for the sequence $\{\mu^{\overline{\gamma}}\}$ of invariant distributions as $\overline{\gamma} \rightarrow \infty$.

\begin{lemma}\label{lem-lowerb}
Let $\Gamma \in \mathcal{B} (\mathbb{S}_{+}^{M})$; then the following lower bound holds:
\begin{equation} \label{eq_lowerb}
\liminf_{\overline{\gamma}\rightarrow\infty}\frac{1}{\overline{\gamma}}\ln\mathbb{\mu}^{\overline{\gamma}}\left(\Gamma^{\circ}\right)\geq
-\inf_{X\in\Gamma^{\circ}}\underline{I}_L(X).
\end{equation}
\end{lemma}
The proof is provided in Appendix~\ref{proof_lem-lowerb}.

\subsection{The LD upper bound}
\label{ld-upper}
In this subsection, we establish the LD upper bound for the family of invariant distributions as $\overline{\gamma} \rightarrow \infty$. The proof is divided into three steps. First, we establish the upper bound on compact sets. Then, we derive a tightness result for the family of invariant distributions. Finally, we establish the LD upper bound on the required closed sets.

First, we provide some basic results on the topological properties of strings.
\begin{definition}[Truncated String]\label{def_truncated}
Let the string $\mathcal{R}$ be given as $\mathcal{R}=(f_{\jmath_1},\cdots,f_{\jmath_r},P_0)$ where $r\in \mathbb{T}_+, \jmath_1,\cdots,{\jmath_r} \in \mathfrak{P} $. Then for $s\leq r$, the truncated string $\mathcal{R}^s$ of length $s$ is defined as
\begin{equation}
\mathcal{R}^s=(f_{\jmath_1},\cdots,f_{\jmath_s},P_0).
\end{equation}
\end{definition}

\begin{lemma}\label{ub_lemma1}
Define the set of strings $\mathcal{U} \subset \mathcal{S}^{P^*}$ and the quantities $l(F)$,
for a closed set $F \in \mathbb{S}_+^{M}$, as
\begin{align}
     \mathcal{U}(F)&=\left\{\mathcal{R}\in \mathcal{S}^{P^*}|\mathcal{N}(\mathcal{R})\in F
     \right\}\\
     l(F)&=\inf_{\mathcal{R}\in \mathcal{U}(F)} \pi(\mathcal{R})\\
     l^{'}(F)&=\inf_{\mathcal{R}\in \mathcal{U}(F)} \overline{w}(\mathcal{R})
\end{align}
where
\begin{equation}
    \overline{w}(\mathcal{R})=\min_{(n_{r},\cdots,n_{1})\in\mathcal{P}_{r}}\sum_{i=1}^{r}\mathbb{I}_{\jmath_{i}\neq
2^{N}-1}\overline{q}_{n_{i}}(\jmath_{i}).
\end{equation}
Then, if $l(F)< \infty$ and $l'(F)< \infty$, there
exists $r_F \in \mathbb{T}_+$ large enough, such that for all $\mathcal{R} \in \mathcal{U}(F)$
with len$(\mathcal{R}) \geq r_F$, we have $\pi(\mathcal{R}^{r_F}) \geq l(F)$ and
$\overline{w}(\mathcal{R}^{r_F}) \geq l'(F)$.
\end{lemma}
In the statement of Lemma~\ref{ub_lemma1}, we assume that the infimum of an empty set is $\infty$. The proof of Lemma~\ref{ub_lemma1} is provided in Appendix~\ref{proof_ub_lemma1}.

From the definition of $\mathcal{U}(F)$, we see that
\begin{equation}
 \mathcal{U}(F)= \bigcup_{X \in F} \mathcal{S}^{P^*} (X)
\end{equation}
and hence
\begin{equation}
l(F)=\inf_{X\in F}\inf_{\mathcal{R}\in \mathcal{S}^{P^*} (X) } \pi(\mathcal{R})=\inf_{X\in F} I(X)
\end{equation}
\begin{equation}
l^{'}(F)=\inf_{X \in F}\inf_{\mathcal{R}\in \mathcal{S}^{P^*} (X) } \overline{w}(\mathcal{R}) =\inf_{X\in F} \overline{I}(X).
\end{equation}

If $l(F) <\infty$, i.e., the set $\mathcal{U}(F)$ is non-empty, the infimum is attained. That is, there exists $\mathcal{R}^* \in \mathcal{U}(F)$ such that
$l(F)=\pi(\mathcal{R}^*)$.

Now we prove the LD upper bound for the family of $\left\{\mu^{\overline{\gamma}}\right\}$ as $\overline{\gamma}\rightarrow \infty$ over compact sets.
\begin{lemma}\label{lem_up_compact}
Let $K \in \mathcal{B}(\mathbb{S}_+^M)$ be a compact set. Then the following upper bound holds:
\begin{equation}
\limsup_{\overline{\gamma}\rightarrow\infty}\frac{1}{\overline{\gamma}}\ln\mathbb{\mu}^{\overline{\gamma}}\left(K\right)\leq
-\inf_{X\in K}\overline{I}_L(X).
\end{equation}
\end{lemma}
The proof is presented in Appendix \ref{proof_lem_up_compact}.

We use the following tightness result to extend the upper bound from compact sets to arbitrary closed sets.
\begin{lemma} \label{lem_up_tightness}
The family of invariant distributions $\left\{\mu^{\overline{\gamma}}\right\}$ satisfies the following tightness property: For every $a>0$, there exists a compact set $K_a \subset \mathbb{S}_+^M$ such that,
\begin{equation}
\limsup_{\overline{\gamma}\rightarrow\infty}\frac{1}{\overline{\gamma}}\ln\mathbb{\mu}^{\overline{\gamma}}\left(K_a^{C}\right)\leq -W(a)
\end{equation}
where
\begin{align}
W(a)=\min_{\mbox{\tiny$\begin{array}{c}
{(n_1,\cdots,n_{{\mbox{\scriptsize len}(\mathcal{R})}})\in \mathcal{P}_{{\mbox{\scriptsize len}(\mathcal{R})}} }\\
{\mathcal{R}:\pi
(\mathcal{R})=\lfloor a\rfloor}\end{array}$}} \sum_{i=1}^{{\mbox{\scriptsize len}(\mathcal{R})}}\mathbb{I}_{\jmath_{i}\neq 2^{N}-1}\overline{q}_{n_{i}}(\jmath_{i}).
\end{align}
\end{lemma}
The proof is presented in Appendix \ref{proof_lem_up_tightness}.

Now we can complete the proof of the LD upper bound for arbitrary closed sets by using the upper bound on compact sets in Lemma~\ref{lem_up_compact} and the tightness result in Lemma~\ref{lem_up_tightness}.
\begin{lemma} \label{lem-upperb}
For a closed set $F\in \mathcal{B}(\mathbb{S}_+^{M})$, the following upper bound holds:
\begin{equation}
\limsup_{\overline{\gamma}\rightarrow\infty}\frac{1}{\overline{\gamma}}\ln\mathbb{\mu}^{\overline{\gamma}}\left(F\right)\leq
-\inf_{X\in F}\overline{I}_L(X).
\end{equation}
\end{lemma}
\begin{proof}
Let $a >0$ be arbitrary. By the tightness estimate in Lemma~\ref{lem_up_tightness}, there exists a compact set $K_a \subset \mathbb{S}_+^M$ such that
\begin{equation} \label{up_final1}
\limsup_{\overline{\gamma}\rightarrow\infty}\frac{1}{\overline{\gamma}}\ln\mathbb{\mu}^{\overline{\gamma}}\left(K_a^{C}\right)\leq -W(a).
\end{equation}
The set $F\cap K_a$, as the intersection of a closed and a compact set, is compact. Then the LD upper bound in Lemma~\ref{lem_up_compact} holds, and we have
\begin{equation}\label{up_final2}
\limsup_{\overline{\gamma}\rightarrow\infty}\frac{1}{\overline{\gamma}}\ln\mathbb{\mu}^{\overline{\gamma}}\left(F\cap K_a\right)\leq
-\inf_{X\in F\cap K_a }\overline{I}_L(X).
\end{equation}
To estimate the probability $\mu^{\overline{\gamma}}(F)$, we use the following decomposition:
\begin{equation}
\mu^{\overline{\gamma}}(F)=\mu^{\overline{\gamma}}(F\cap K_a)+\mu^{\overline{\gamma}}(F\cap K_a^C)\leq \mu^{\overline{\gamma}}(F\cap K_a)+ \mu^{\overline{\gamma}}( K_a^C).
\end{equation}
From the results on the limits of real number sequences (see Lemma 1.2.15 of \cite{DemboZeitouni}), we have
\begin{align}
&\limsup_{\overline{\gamma}\rightarrow\infty}\frac{1}{\overline{\gamma}}\ln\mathbb{\mu}^{\overline{\gamma}}\left(F\right) \leq  \nonumber\\
&\max \left( \limsup_{\overline{\gamma}\rightarrow\infty}\frac{1}{\overline{\gamma}}\ln\mathbb{\mu}^{\overline{\gamma}}\left(F\cap K_a \right), \limsup_{\overline{\gamma}\rightarrow\infty}\frac{1}{\overline{\gamma}}\ln\mathbb{\mu}^{\overline{\gamma}}\left(K_a^C\right) \right).\nonumber
\end{align}

From (\ref{up_final1}) and (\ref{up_final2}), we have
\begin{align}
\limsup_{\overline{\gamma}\rightarrow\infty}\frac{1}{\overline{\gamma}}\ln\mathbb{\mu}^{\overline{\gamma}}\left(F\right) &\leq \max \left( -\inf_{X\in F\cap K_a }\overline{I}_L(X),-W(a)\right)\nonumber\\
&\leq  \max \left( -\inf_{X\in F }\overline{I}_L(X),-W(a)\right)\nonumber\\
&=-\min \left( \inf_{X\in F }\overline{I}_L(X),W(a)\right).\nonumber
\end{align}

Since the above inequality holds for an arbitrary $a>0$, taking the limit as $a\rightarrow \infty$ on both sides together with $W(a)\rightarrow \infty$, we have
\begin{align}
\limsup_{\overline{\gamma}\rightarrow\infty}\frac{1}{\overline{\gamma}}\ln\mathbb{\mu}^{\overline{\gamma}}\left(F\right) &\leq \max \left( -\inf_{X\in F\cap K_a }\overline{I}_L(X),-W(a)\right)\nonumber\\
&\leq  -\inf_{X\in F }\overline{I}_L(X).
\end{align}
\end{proof}

\section{Proof of Theorem~\ref{th:ldp}}\label{proof_th:ldp}
We are now ready to complete the proof of Theorem~\ref{th:ldp}.
\begin{proof}
Lemma~\ref{lem-lowerb} and Lemma~\ref{lem-upperb} have established that the family of $\left\{\mu^{\overline{\gamma}}\right\}$ satisfies the LD lower and upper bounds at scale $\overline{\gamma}$ with rate functions $\underline{I}_L$ and $\overline{I}_L$, respectively, as $\overline{\gamma}\rightarrow \infty$. To complete the proof of Theorem~\ref{th:ldp}
it suffices to show that $\overline{I}_L(\cdot)=\overline{I}(\cdot)$ and $\underline{I}_L(\cdot)=\underline{I}(\cdot)$, i.e., $\overline{I}(\cdot)$ and $\underline{I}(\cdot)$ are lower semicontinuous. We first prove  $\overline{I}_L(\cdot)=\overline{I}(\cdot)$, and it takes the same method to prove $\underline{I}_L(\cdot)=\underline{I}(\cdot)$.

If $\overline{I}_L(X)=\infty$, from Proposition~\ref{prop_rate_function} (iii), clearly $\overline{I}_L(X)=\overline{I}(X),~\forall X \in \mathbb{S}_+^M$. Then, we consider the case $\overline{I}_L(X)<\infty$. From the definition
\begin{equation}
\overline{I}_L(X)= \lim_{\varepsilon \rightarrow \infty} \inf_{Y\in B_{\varepsilon}(X)} \overline{I}(Y),
\end{equation}
we know the discrete quantity $ \inf_{Y\in B_{\varepsilon}(X)} \overline{I}(Y)$ is non-decreasing w.r.t. $\varepsilon$; then there exists $\varepsilon_0 >0$ such that
\begin{equation}
\overline{I}_L(X)= \inf_{Y\in B_{\varepsilon}(X)} \overline{I}(Y),~\forall \varepsilon \leq \varepsilon_0.
\end{equation}
The infimum above is achieved for every $\varepsilon _0 > 0$, and we conclude that there exists a sequence $\{X_n\}_{n\in \mathbb{N}}$ such that
\begin{equation}
X_n \in \overline{B_{\varepsilon_0}}(X), ~ \lim_{n\rightarrow \infty} X_n=X,~ \overline{I}(X_n)=\overline{I}_L(X).
\end{equation}
Recall the set of strings
\begin{equation}
\mathcal{U}( \overline{B_{\varepsilon_0}}(X) )= \{ \mathcal{R} \in \mathcal{S}^{P^*}|\mathcal{N}(\mathcal{R})\in \overline{B_{\varepsilon_0}}(X)  \}.
\end{equation}
Then we have
\begin{equation}
l' (\overline{B_{\varepsilon_0}}(X))= \inf_{Y\in \overline{B_{\varepsilon_0}}(X)} \overline{I}(Y) = \overline{I}_L(X).
\end{equation}

Since $\overline{B_{\varepsilon_0}}(X)$ is closed, by Lemma~\ref{ub_lemma1}, there exists $r_0 \in \mathbb{T}_+$ such that for $\mathcal{R} \in \mathcal{U}( \overline{B_{\varepsilon_0}}(X) )$ with len$(\mathcal{R})\geq r_0$,
\begin{equation}
  \overline{w}(\mathcal{R}^{r_0}) \geq l'(\overline{B_{\varepsilon_0}}(X))= \overline{I}_L(X).
\end{equation}

By the existence of $\{X_n\}$, there exists a sequence $\{\mathcal{R}_n \}$ of strings in $\mathcal{U}( \overline{B_{\varepsilon_0}}(X) )$ such that
\begin{equation}\label{eq_pfmain}
\mathcal{N}(\mathcal{R}_n)=X_n,~ \overline{w}(\mathcal{R}_n)=  \overline{I}_L(X).
\end{equation}

Without loss of generality, we assume that len$(\mathcal{R}_n)=r_0$ for all $n$. Indeed, if len$(\mathcal{R}_n) <r_0$, we can modify $\mathcal{R}_n$ by appending the requisite number of $f_{2^N-1}$ at the right end, which still satisfies (\ref{eq_pfmain}). On the other hand, if len$(\mathcal{R}_n) >r_0$, we note that $\mathcal{R}_n$ must be of the form
\begin{equation}
\mathcal{R}_n=\left( f_{\jmath_1},\cdots,f_{\jmath_{r_0}}, f_{2^N-1}^{\mbox{\scriptsize len} (\mathcal{R}_n)-r_0}, P^* \right)
\end{equation}
where the truncated string (defined in Definition~\ref{def_truncated}) $\mathcal{R}_n^{r_0}$ satisfies $\mathcal{N}(\mathcal{R}_n^{r_0})=X_n$ and $\overline{w}(\mathcal{R}_n^{r_0})=\overline{I}_L(X)$. Hence, if len$(\mathcal{R}_n) >r_0$, we may consider the truncated string $\mathcal{R}_n^{r_0}$ instead, which also satisfies (\ref{eq_pfmain}). We thus assume that the sequence $\{ \mathcal{R}_n \}$ with the properties in (\ref{eq_pfmain}) further satisfies len$(\mathcal{R}_n)=r_0$ for all $n$.

The number of distinct strings in the sequence $\{ \mathcal{R}_n \}$ is at most $(2^N-1)^{r_0}$; in fact, it should be less than $(2^N-1)^{r_0}$ due to the constraint $ \overline{w}(\mathcal{R}_n)=\overline{I}_L(X)$. Hence, at least one pattern is repeated infinitely often in the sequence $\{ \mathcal{R}_n \}$, i.e., there exists a string $ \mathcal{R}^* $ such that we have len$(\mathcal{R}^*)=r_0,~\overline{w}(\mathcal{R}^*)=  \overline{I}_L(X)$, and a subsequence $\{ \mathcal{R}_{n_k}\}_{k \in \mathbb{N}}$ of $\{ \mathcal{R}_{n}\}$ with $\mathcal{R}_{n_k}=\mathcal{R}^*$.

The corresponding subsequence $\{X_{n_k}\}$ of numerical values then satisfies
\begin{equation}
X_{n_k}= \mathcal{N} ( \mathcal{R}_{n_k})= \mathcal{N}(\mathcal{R}^*),~\forall k\in \mathbb{N},
\end{equation}
and hence we have
\begin{equation}
X=\lim_{k\rightarrow \infty} X_{n_k}=\mathcal{N}(\mathcal{R}^*).
\end{equation}

Therefore, we have the string $\mathcal{R}^* \in \mathcal{S}^{P^*}(X)$ and
\begin{equation}
\overline{I}(X)= \inf_{\mathcal{R} \in \mathcal{S}^{P^*}(X)} \overline{w}(\mathcal{R}) \leq \overline{w}(\mathcal{R}^*)=\overline{I}_L(X).
\end{equation}

With the fact that $\overline{I}(X) \geq \overline{I}_L(X)$, we have the final conclusion:
\begin{equation}
\overline{I}(X)=\overline{I}_L(X).
\end{equation}
\end{proof}
\section{Simulation Results}\label{simulation}
\begin{figure}
\begin{center}
\subfigure[]{
\includegraphics[width=0.4\textwidth]{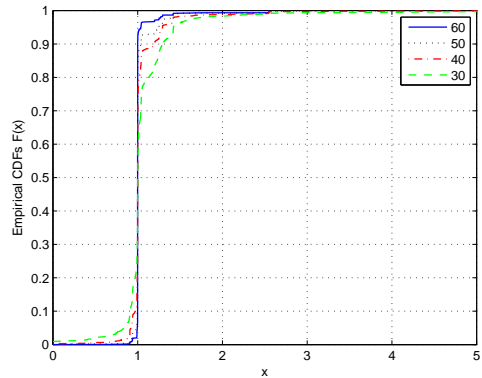}}\\
\subfigure[]{
\includegraphics[width=0.4\textwidth]{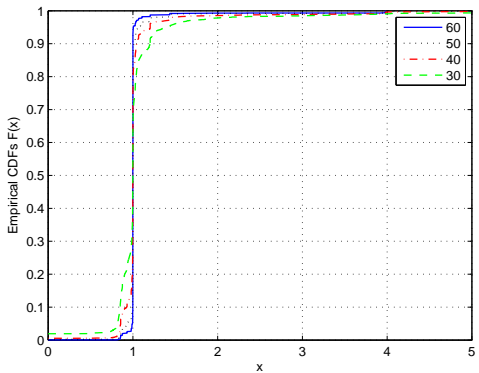}
}
\caption{(a): CDF of the normalized largest eigenvalue from $\mu^{\overline{\gamma}}$ for varying $\overline{\gamma}=30,40,50,60$. (b): CDF of the normalized trace from $\mu^{\overline{\gamma}}$ for varying $\overline{\gamma}=30,40,50,60$} \label{CDF_figure}
\end{center}
\end{figure}
\begin{figure}
\begin{center}
\subfigure[]{
\includegraphics[width=0.4\textwidth]{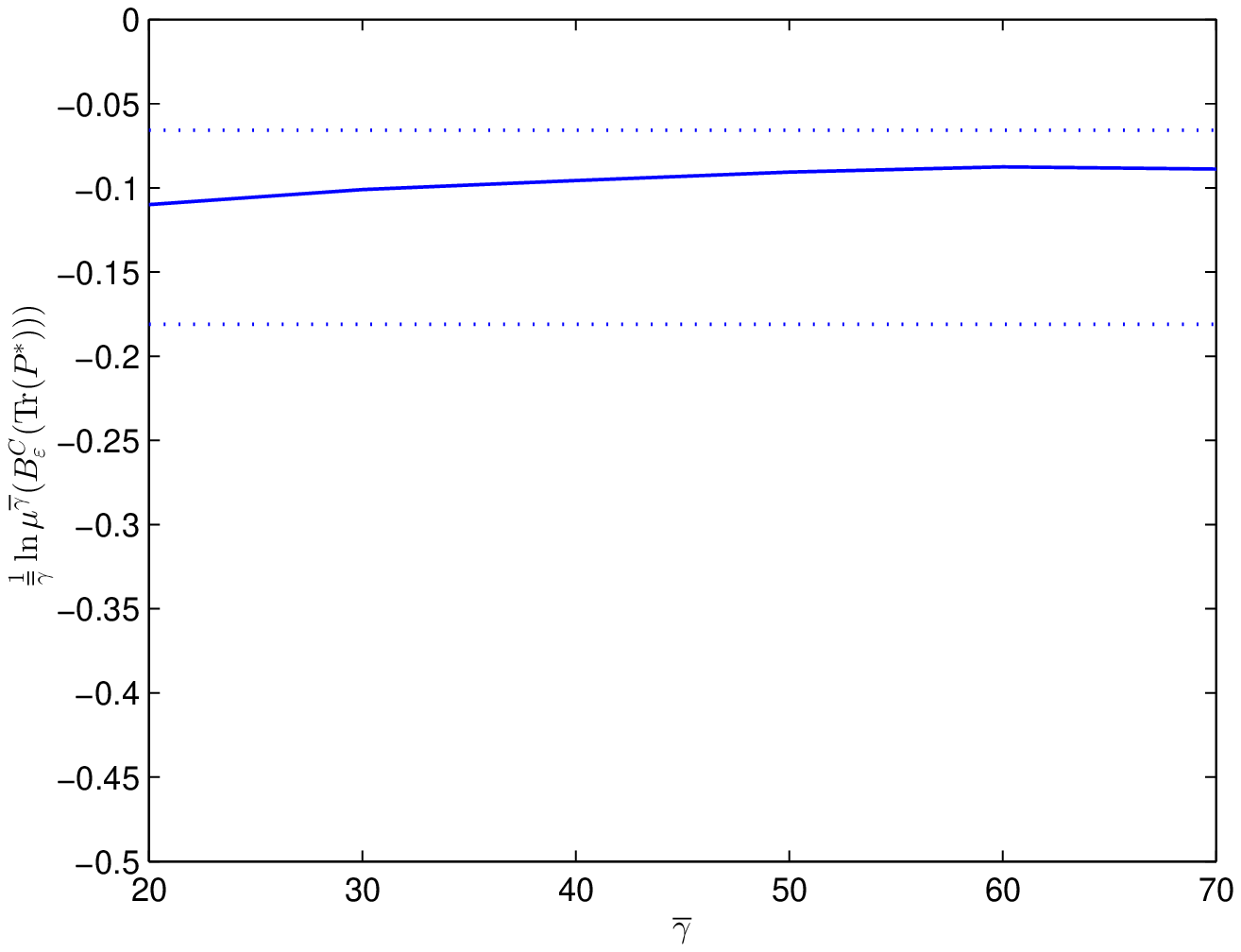}}\\
\subfigure[]{
\includegraphics[width=0.4\textwidth]{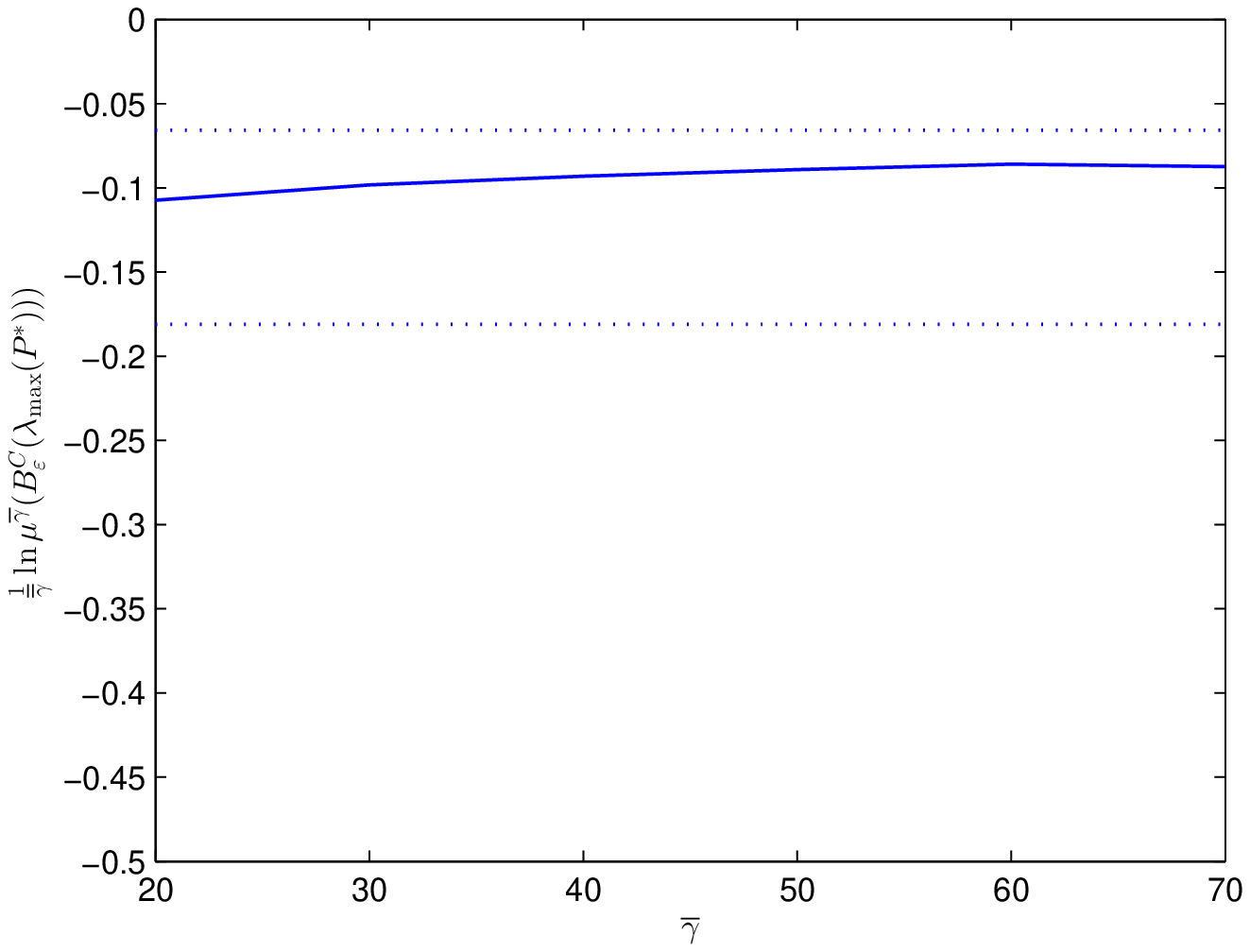}
}
\caption{(a): LD decay exponent for probability of rare event $B_\varepsilon^C(\mbox{Tr}(P^*))$ and the LD upper and lower bounds. (b): LD decay exponent for probability of rare event $B_{\varepsilon'}^C(\lambda_{\max}(P^*))$ and the LD upper and lower bounds} \label{Decay_figure}
\end{center}
\end{figure}
In this section, we simulate the M-GIKF to estimate a $10$-dimensional state-unknown system\footnote{We acknowledge that this is not a large system size; it just illustrates the concept.} with a network of $5$ sensors. The matrices $\mathcal{F},~\mathcal{C}_n$, and $\mathcal{Q}$ satisfy Assumptions S.1 and D.1. The simulation is based on the example of distributed observation dissemination protocol discussed in Section~\ref{ex-dist}, in which this protocol does not use knowledge of global topology and is the simplest random walk on the graph with uniform (unoptimized) neighbor selection. By tuning the link selection probabilities (using full knowledge of global topology), it could be possible to perform better. The protocol in Section~\ref{ex-dist} is just an example of possible protocols, while the theoretical analysis in our paper is protocol independent.

We study the behavior of $\mu^{\overline{\gamma}}$ for different values of $\overline{\gamma}$. We iterate the RRE $10^4$ times to ensure the error covariance sequence at a randomly selected sensor converged in distribution to $\mu^{\overline{\gamma}}$ as shown in Theorems~\ref{main1} and~\ref{main2}, where we simulate $5,000$ samples for each $\overline{\gamma}$. In order to graphically present the distribution for the covariance matrix, we focus on its largest eigenvalue and trace here. The resulting empirical Cumulative Distribution Functions (CDFs) of the normalized largest eigenvalue $\lambda_{\max}(\cdot)$ and the normalized trace$\mbox{ Tr}(\cdot)$ (which is the conditional mean-squared error) of the error covariance matrices are plotted in Fig.~\ref{CDF_figure}, where the x-axis is $\lambda_{\max}(\cdot)/\lambda_{\max}(P^*)$ and $\mbox{Tr}(\cdot)/\mbox{Tr}(P^*)$, respectively in Fig.~\ref{CDF_figure} (a) and Fig.~\ref{CDF_figure} (b). As $\overline{\gamma}$ increases, we see that the empirical measure $\mu^{\overline{\gamma}}$ converges in distribution to the Dirac measure $\delta_{P^*}$ of $P^*$.

Then we simulate the LD decay exponent of the rare event $B_\varepsilon^C(\mbox{Tr}(P^*))$ with $\varepsilon=\mbox{Tr}(P^*)/2$, and the LD decay exponent of the rare event $B_{\varepsilon'}^C(\lambda_{\max}(P^*))$ with $\varepsilon'=\lambda_{\max}(P^*)/2$. For each $\overline{\gamma}$, we estimate the LD decay exponents $\frac{1}{\overline{\gamma}}\ln\mathbb{\mu}^{\overline{\gamma}}(B_\varepsilon^C(\lambda_{\max}(P^*)))$ and $\frac{1}{\overline{\gamma}}\ln\mathbb{\mu}^{\overline{\gamma}}(B_\varepsilon^C(\mbox{Tr}(P^*)))$ by using the samples obtained above for calculating the empirical CDFs. Then we take effort to numerically calculate the LD lower and upper bounds. From Theorem~\ref{th:ldp}, the LD upper bound for the rare event $B_\varepsilon^C(\mbox{Tr}(P^*))$ can be obtained as the negative infimum of $\overline{I}(\cdot)$ over the set of rare events, and the LD lower bound for the rare event $B_\varepsilon^C(\mbox{Tr}(P^*))$ can be obtained as the negative infimum of $\underline{I}(\cdot)$ over the set of rare events. Recall \eqref{up-weight}, \eqref{up-rate}, and \eqref{th:ldp1}, we present the LD upper bound for the rare event $B_\varepsilon^C(\mbox{Tr}(P^*))$ as
\begin{align}\label{simu_1}
-\inf_{\mbox{\tiny Tr}(X)\in{\overline{B_\varepsilon^C(\mbox{\tiny Tr}(P^*))}}}&\inf_{\mathcal{R}\in\mathcal{S}^{P^{\ast}}(X)} \min_{(n_{\mbox{\tiny len}(\mathcal{R})},\cdots,n_{1})\in\mathcal{P}_{\mbox{\tiny len}(\mathcal{R})}}\nonumber\\
&\sum_{i=1}^{{\mbox{\tiny len}(\mathcal{R})}}\mathbb{I}_{\jmath_{i}\neq 2^{N}-1}\overline{q}_{n_{i}}(\jmath_{i}),~X \in \mathbb{S}_+^{M},
\end{align}
where $\overline{q}_{n_{i}}(\jmath_{i})$ is defined in \eqref{alpbeta}. Now we present one way to set $\alpha$ in \eqref{alpbeta}. Recall Section~\ref{ex-dist}, where $T_{i}$ is the hitting time starting from sensor $i$ to another particular sensor $n$ in the Markov chain with transition matrix $\overline{A^o}$. Then we have
\begin{equation}
P(T_i>L)=\sum_{n_1,\cdots,n_L\neq n}\overline{A^o}_{n_1n_2}\overline{A^o}_{n_2n_3}\cdots\overline{A^o}_{n_{L-1}n_L},
\end{equation}
and $\alpha$ can be selected as $\alpha=\max_{i}P(T_i>L)$.

For the LD lower bound of the rare event $B_\varepsilon^C(\mbox{Tr}(P^*))$, recall \eqref{up-weight1}, \eqref{up-rate}, and \eqref{th:ldp2}, we present the LD lower bound as
\begin{align}\label{simu_2}
-\inf_{\mbox{\tiny Tr}(X)\in{{B_\varepsilon^C(\mbox{\tiny Tr}(P^*))}^o}}&\inf_{\mathcal{R}\in\mathcal{S}^{P^{\ast}}(X)} \min_{(n_{\mbox{\tiny len}(\mathcal{R})},\cdots,n_{1})\in\mathcal{P}_{\mbox{\tiny len}(\mathcal{R})}}\nonumber\\
&\sum_{i=1}^{{\mbox{\tiny len}(\mathcal{R})}}\mathbb{I}_{\jmath_{i}\neq 2^{N}-1}\underline{q}_{n_{i}}(\jmath_{i}), ~X \in \mathbb{S}_+^{M},
\end{align}
where $\underline{q}_{n_{i}}(\jmath_{i})$ is defined in \eqref{alpbeta}. We could set $\beta$ in \eqref{alpbeta} as $\beta=\min_{i}P(T_i>L)$.

Then the problems of computing the LD upper and lower bounds could be converted to solving the optimization problems in \eqref{simu_1} and \eqref{simu_2} respectively. We could then apply some search method to numerically solve those problems to obtain the LD upper and lower bounds. The same analysis could be applied  for the case of $B_\varepsilon^C(\lambda_{\max}(P^*))$.

Fig.~\ref{Decay_figure} displays the estimated LD decay exponents of the rare events of $B_\varepsilon^C(\mbox{Tr}(P^*))$ and $B_{\varepsilon'}^C(\lambda_{\max}(P^*))$ for different values of $\overline{\gamma}$, and the corresponding LD upper and lower bounds of the decay exponents. The empirically estimated decay exponents in these two rare events perform quite similar, which is due to the fact that $\varepsilon$ and $\varepsilon'$ have the same relative factor $0.5$ for the maximum eigenvalue and the trace of $P^*$, respectively.

Finally, note that the convergence rate with respect to $\overline{\gamma}$ (i.e., the large deviation exponent) may be improved by considering a more sophisticated observation dissemination protocol. For instance, the neighbor selection probabilities in the observation dissemination protocol from Section~\ref{ex-dist} may be optimized for a given communication network structure. This could lead to a faster mixing Markov chain governing the observation dissemination and, hence, for the same rate $\overline{\gamma}$, a sensor could more likely receive the observations of more sensors in each epoch (also see Remark~\ref{rem:obs}), leading to faster convergence of $\{\mathbb{\mu}^{\overline{\gamma}}\}$ to $\delta_{P^{\ast}}$.
\section{Conclusions} \label{conclusion}
We have proposed a distributed Kalman filtering scheme, the Modified Gossip Interactive Kalman Filter (M-GIKF), where the filtered states are exchanged and the observations are propagated at a rate $\overline{\gamma}$ among the sensors over the network graph. We have shown that, for each $\overline{\gamma}>0$, the conditional estimation error covariance at a randomly selected sensor converges weakly (in distribution) to a unique invariant measure $\mu^{\overline{\gamma}}$ of an associated RRE. To prove this, we have interpreted the filtered states as stochastic particles and formulated the resulting random Riccati equation with Markov modulated switching as a random dynamical system. After establishing weak convergence to $\mu^{\overline{\gamma}}$, we have further characterized $\mu^{\overline{\gamma}}$ as $\overline{\gamma} \rightarrow \infty$, showing that $\mu^{\overline{\gamma}}$ satisfies the large deviation upper and lower bounds, providing a tradeoff between communication rate and estimation accuracy. In particular, we have shown that the distributed M-GIKF approaches the centralized performance exponentially fast in $\overline{\gamma}$, the communication rate parameter, in that the measure $\mathbb{\mu}^{\overline{\gamma}}$ converges to $\delta_{P^{\ast}}$ exponentially fast in probability as $\overline{\gamma}\rightarrow\infty$.

\appendices

\section{Proof of Theorem~\ref{thm_convDirac}} \label{proof_convDiarac}
First define the following class of sets:
\begin{equation}
C=\{ \mathcal{F} | \mathcal{F} ~\mbox{is~closed~and}~P^* \in \mathcal{F}\}.
\end{equation}
Then proving this theorem is equivalent to proving the following:
\begin{equation}
\lim_{\overline{\gamma} \rightarrow \infty} d_P(\mu^{\overline{\gamma}},\delta_{P^*}) =0
\end{equation}
where the Prohorov metric $d_P(\mu^{\overline{\gamma}},\delta_{P^*})=\inf \{ \varepsilon >0 ~| ~\mu^{\overline{\gamma}}(\mathcal{F}_\varepsilon) + \varepsilon \geq 1,~\forall \mathcal{F} \in C \}$ is defined in \cite{Jacod-Shiryaev}.

Consider $0 < \varepsilon < 1$ small enough. Then there exists a $\varepsilon_0 > 0$ such that, for every $\mathcal{F} \in C$, we have $B_{\varepsilon_0} (P^*) \subset \mathcal{F}_{\varepsilon} $. The numerical value of the string $\mathcal{R}=P^*$ belongs to $B_{\varepsilon_0}(P^*)$, and hence by (\ref{lb_4}), there exists an integer $r^*$ such that
\begin{align}
  \mu^{\overline{\gamma}}(B_{\varepsilon_0} (P^*))&\geq \prod_{k=1}^{r_0}q_n(\jmath_k)\prod_{k=r_{0}+1}^{r_{0}+r_{\varepsilon_1}}q_{n_k}(2^N-1) \nonumber\\
   &= \prod_{k=1}^{r^*}q_{n_k}(2^N-1),  ~~\mbox{where}~ r^*=r_{0}+r_{\varepsilon_1}.\nonumber
\end{align}
Thus, for all $\mathcal{F} \in C$, we have
\begin{equation}
\mu^{\overline{\gamma}} (\mathcal{F}_{\varepsilon})\geq \mu^{\overline{\gamma}}(B_{\varepsilon_0} (P^*)) \geq \prod_{k=1}^{r^*}q_{n_k}(2^N-1). 
\end{equation}

Since $q_{n_k}(2^N-1) \rightarrow 1$ as $\overline{\gamma} \rightarrow \infty$, we have for $\overline{\gamma} \rightarrow \infty$
\begin{equation}
\mu^{\overline{\gamma}}(\mathcal{F}_\varepsilon) + \varepsilon \geq \prod_{k=1}^{r^*}q_{n_k}(2^N-1)+ \varepsilon \geq 1.
\end{equation}

Then, following the definition of $d_P(\mu^{\overline{\gamma}},\delta_{P^*})$, when $\overline{\gamma} \rightarrow \infty$, we have
\begin{equation}
d_P(\mu^{\overline{\gamma}},\delta_{P^*}) \leq \varepsilon,~ \overline{\gamma} \rightarrow \infty.
\end{equation}
Hence,
\begin{equation}
\lim_{\overline{\gamma} \rightarrow \infty} d_P(\mu^{\overline{\gamma}},\delta_{P^*}) \leq \varepsilon.
\end{equation}
Since $\varepsilon >0$ is arbitrary, by considering the limit as $\varepsilon \rightarrow 0$, we conclude that $\lim_{\overline{\gamma} \rightarrow \infty} d_P(\mu^{\overline{\gamma}},\delta_{P^*})=0$.
\section{Proof of Lemma~\ref{lem-lowerb}} \label{proof_lem-lowerb}
Since the sequence $\{P_{n}(k)\}$ converges weakly (in distribution) to $\mu^{\bar{\gamma}}$, we have
\begin{equation} \label{lb_1}
    \limsup_{k\rightarrow \infty} \mathbb{P}\!\left(P_{n}(k)\in F\right) \!\leq\!
    \mu^{\overline{\gamma}}(F),\!~\forall ~\mbox{closed set space}~ F \!\subset\! \mathbb{S}_{+}^M.
\end{equation}

Consider a measurable set $\Gamma\in \mathcal{B}(\mathbb{S}_{+}^M)$. Note that if $\Gamma$ has an
empty interior $\Gamma^{\circ}$, the assertion in \eqref{eq_lowerb} holds trivially since the right-hand side becomes $-\infty$. We thus consider the non-trivial case in which $\Gamma^{\circ}\neq\emptyset$. Let $X \in \Gamma^\circ \cap \mathcal{D}_{\underline{I}}$, with $\mathcal{D}_{\underline{I}}$ as the
effective domain of $\underline{I}(\cdot)$, i.e., the set on which $\underline{I}(\cdot)$ is
finite. There exists a small enough $\varepsilon > 0$, such that the closed ball
$\overline{B}_{\varepsilon}(X) \in \Gamma^{\circ}$. Then, from (\ref{lb_1}), we have
\begin{equation}\label{lb_2}
    \mu^{\overline{\gamma}}(\Gamma^\circ)\geq\mu^{\overline{\gamma}}\left(\overline{B}_\varepsilon(X)\right)\geq\limsup_{t\rightarrow
    \infty} \mathbb{P}\left(P_{n}(k)\in \overline{B}_\varepsilon(X) \right).
\end{equation}

Now we calculate the right-hand side of (\ref{lb_2}). The set $\mathcal{S}^{P^*}(X)$ is non-empty, due to
the fact that $X\in\mathcal{D}_{\underline{I}}$ implying that $\underline{I}(X)$ is finite and $\mathcal{S}^{P^*}(X)$ is non-empty. Hence, for some $r_0 \in \mathbb{T}_+$ and
$\jmath_1,...,\jmath_{r_0}\in \mathfrak{P}$,~ we have a string
$\mathcal{R}=(f_{\jmath_1},...,f_{\jmath_{r_0}},P^*)\in \mathcal{S}^{P^*}(X)$. Define the function
$g$: $\mathbb{S}^M_+\mapsto\mathbb{S}^M_+$ by $g(Y)=f_{\jmath_1}\circ ... \circ
f_{\jmath_{r_0}}(Y)$. Since $g$ is continuous, there exists $\varepsilon_1 > 0$ such that
\begin{equation}
    \left\|g(Y)-g(P^*)\right\|\leq \varepsilon, \forall ~Y\in \overline{B}_{\varepsilon_1}(P^*).
\end{equation}

With Proposition~\ref{unif_conv} (ii), for $\varepsilon_1 >0$, there exists $r_{\varepsilon_1}$ such that
\begin{equation}
    \left\|f_{2^N-1}^r(Y)-P^*\right\|\leq \varepsilon_1, \forall~r\geq r_{\varepsilon_1},~Y\in \mathbb{S}^M_+.
\end{equation}
For any $r\in \mathbb{T}_+$ such that $r\geq r_0+r_{\varepsilon_1}$ and any string $\mathcal{R}_1 \in
\mathcal{S}^{P_0}_r$ of the form
\begin{equation}
    \mathcal{R}_1 = \left\{f_{\jmath_1},...,f_{\jmath_{r_0}},f_{2^N-1}^{r_{\varepsilon_1}},f_{i_1},...,f_{i_{r-r_0-r_{\varepsilon_1}}},
    P_0\right\}\nonumber
\end{equation}
where $f_{i_1},...,f_{i_{r-r_0-r_{\varepsilon_1}}} \in \mathfrak{P}$, it follows that
\begin{align}
    &\left\|\mathcal{N}(\mathcal{R}_1)-X\right\|=\left\|\mathcal{N}(\mathcal{R}_1)-\mathcal{N}(\mathcal{R})\right\|\nonumber\\
    &=\left\|g\left(f_{2^N-1}^{r_{\varepsilon_1}}(f_{i_1},...,f_{i_{r-r_0-r_{\varepsilon_1}}}(P_0))\right)-g(P^*)\right\|\leq \varepsilon,\nonumber
\end{align}
which is derived from the fact that
\begin{equation}
    \left\|f_{2^N-1}^ {r_{\varepsilon_1}}(f_{i_1},...,f_{i_{r-r_0-r_{\varepsilon_1}}}(P_0))-P^*\right\|\leq \varepsilon_1.\nonumber
\end{equation}
Therefore,
\begin{equation}
    \mathcal{N}(\mathcal{R}_1)\in \overline{B}_\varepsilon(X).\nonumber
\end{equation}
For $r\geq {r_0+r_{\varepsilon_1}}$, define the set of strings
\begin{align}
    \mathcal{R}_t=&\left\{\left.\left( f_{\jmath_1},...,f_{\jmath_{r_0}},f_{2^N-1}^{r_{\varepsilon_1}},f_{i_1},...,f_{i_{r-r_0-r_{\varepsilon_1}}},
    P_0 \right)\right| \right.\nonumber\\
    &~~\left. f_{i_1},...,f_{i_{r-r_0-r_{\varepsilon_1}}} \in \mathfrak{P} \right\}.
\end{align}
Then, it follows that $ \mathcal{N}(\mathcal{R}_2)\in \overline{B}_\varepsilon(X),
\forall~\mathcal{R}_2\in\mathcal{R}_t$. Thus, for $r\geq
r_0+r_{\varepsilon_1}$, we have
\begin{align}\label{lb_3}
    &\mathbb{P}\left(P_{n}(k)\in \overline{B}_\varepsilon(X)\right)\geq
    \mathbb{P}\left(P_{n}(k)\in \mathcal{N}(\mathcal{R}_t)\right)\nonumber\\
    &=\sum_{{i_1},...,{i_{r-r_0-r_{\varepsilon_1}}} \in \mathfrak{P}}
    \left[\prod_{k=1}^{r_0}q_{n_k}(\jmath_k)\right]
    \left[\prod_{k=r_{0}+1}^{r_{0}+r_{\varepsilon_1}}{q_{n_k}(2^N-1)}\right]\nonumber\\
    &~~~~\left[\prod_{k'=1}^{r-r_0-r_{\varepsilon_1}}q_{n_{k'}}(i_{k'})\right]\nonumber\\
    &=\prod_{k=1}^{r_0}q_n(\jmath_k)\prod_{k=r_{0}+1}^{r_{0}+r_{\varepsilon_1}}q_{n_k}(2^N-1).
\end{align}

From (\ref{lb_2}) and (\ref{lb_3}), there exists
\begin{equation}\label{lb_4}
  \mu^{\overline{\gamma}}(\Gamma^\circ)\geq
\prod_{k=1}^{r_0}q_n(\jmath_k)\prod_{k=r_{0}+1}^{r_{0}+r_{\varepsilon_1}}q_{n_k}(2^N-1)
\end{equation}
and hence
\begin{align}
    \ln \mu^{\overline{\gamma}}(\Gamma^\circ) & \geq \sum_{k=1}^{r_0} \mathbb{I}_{\jmath_k\neq {2^N-1}} \ln
    q_{n_k}(\jmath_k)+ \sum_{k=1}^{r_0} \mathbb{I}_{\jmath_k = {2^N-1}} \ln
    q_{n_k}(\jmath_k)\nonumber\\
&~~~+\sum_{k=r_{0}+1}^{r_{0}+r_{\varepsilon_1}}\ln q_{n_k}(2^N-1) .
\end{align}
Since $\lim_{\overline{\gamma} \rightarrow \infty} q_{n_k}(2^N-1) =1 $, i.e., the probability of each sensor obtaining the full set of observations through the observation dissemination protocol approaches 1 as the communication rate $\overline{\gamma}\rightarrow\infty$, we have
\begin{align}
    \liminf_{\overline{\gamma} \rightarrow \infty} \frac{\ln \mu^{\overline{\gamma}}(\Gamma^\circ)}{\overline{\gamma}} &\geq \liminf_{\overline{\gamma} \rightarrow \infty} \sum_{k=1}^{r_0} \mathbb{I}_{\jmath_k\neq
{2^N-1}} \frac{1}{\overline{\gamma}}\ln q_{n_k}(\jmath_k)\nonumber\\
&\geq \sum_{k=1}^{r_0} \mathbb{I}_{\jmath_k\neq {2^N-1}} \liminf_{\overline{\gamma} \rightarrow
\infty} \frac{1}{\overline{\gamma}}\ln q_{n_k}(\jmath_k) \nonumber\\
&\geq -\sum_{k=1}^{r_0}
\mathbb{I}_{\jmath_k\neq {2^N-1}} \underline{q}_{n_k}(\jmath_k)
\end{align}
where the last inequality follows from the fact that $\liminf_{\overline{\gamma} \rightarrow
\infty} \frac{1}{\overline{\gamma}}\ln q_{n_k}(\jmath_k)\geq  \underline{q}_{n_k}(\jmath_k)$.

Since the above holds for all $(n_{r_0},\cdots,n_{1})\in\mathcal{P}_{r_0}$, we have
\begin{align}
   \liminf_{\overline{\gamma} \rightarrow \infty} \frac{\ln
   \mu^{\overline{\gamma}}(\Gamma^\circ)}{\overline{\gamma}} &\!\geq \!\max_{(n_{r_0},\cdots,n_{1})\in\mathcal{P}_{r_0}} \!\left\{-\sum_{k=1}^{r_0}
\mathbb{I}_{\jmath_k\neq {2^N-1}} \underline{q}_{n_k}(\jmath_k) \right\}\nonumber\\
&=-\underline{w}(\mathcal{R})
\end{align}
with
$\underline{w}(\mathcal{R})=\min_{(n_{r_0},\cdots,n_{1})\in\mathcal{P}_{r_0}}\sum_{i=1}^{r_0}\mathbb{I}_{\jmath_{i}\neq
2^{N}-1}\underline{q}_{n_{i}}(\jmath_{i})$.

Given that the above holds for all $\mathcal{R} \in \mathcal{S}^{P^*}(X)$, we have
\begin{align}
    \liminf_{\overline{\gamma} \rightarrow \infty} \frac{\ln
   \mu^{\overline{\gamma}}(\Gamma^\circ)}{\overline{\gamma}} &\geq \sup_{\mathcal{R} \in
   \mathcal{S}^{P^*}(X)}{(-\underline{w}(\mathcal{R}))}\nonumber\\
   &=-\inf_{{\mathcal{R} \in
   \mathcal{S}^{P^*}(X)}}{\underline{w}(\mathcal{R})}=-\underline{I}(X).
\end{align}
Finally, from the fact that for $X \notin D_{\underline{I}} $, $\underline{I}(X)=\infty$, we have
\begin{equation}
\liminf_{\overline{\gamma} \rightarrow \infty} \frac{\ln
   \mu^{\overline{\gamma}}(\Gamma^\circ)}{\overline{\gamma}} \geq -\inf_{X \in \Gamma^\circ \cap
\mathcal{D}_{\underline{I}} } \underline{I}(X)= -\inf_{X \in \Gamma^\circ } \underline{I}(X).\nonumber
\end{equation}

Since $\Gamma^\circ$ is open, from Proposition~\ref{prop_rate_function} (iii), we have
\begin{equation}
-\inf_{X \in \Gamma^\circ } \underline{I}_L(X)=-\inf_{X \in \Gamma^\circ } \underline{I}(X).
\end{equation}
Thus, the proof is completed.

\section{Proof of Lemma~\ref{ub_lemma1}} \label{proof_ub_lemma1}
We first prove, if $l'(F)< \infty$, there
exists $r_F \in \mathbb{T}_+$ large enough, such that for all $\mathcal{R} \in \mathcal{U}(F)$
with len$(\mathcal{R}) \geq r_F$, we have $\overline{w}(\mathcal{R}^{r_F}) \geq l'(F)$. Then the proof for the case of $l(F)$ naturally follows.

The case $l'(F)=0$ is trivial, by choosing an arbitrary positive $r_F$.
Consider the case $l'(F)\geq \overline{q}$, where $\overline{q}=\min_{1\leq n\leq N, \jmath \in \mathfrak{P}} \overline{q}_n(\jmath)$. Using an inductive argument, it suffices to show that for every $\overline{q} \leq i \leq l'(F)$, there exists a positive $r_F^i \in \mathbb{T}_+$ such that, for $\mathcal{R} \in \mathcal{U}(F)$ with $\mbox{len}(\mathcal{R}) \geq  r_F^i$, we have
\begin{equation} \label{eqn1_up_lm1pf}
\overline{w}\left( \mathcal{R}^{r_F^i}\right)\geq i.
\end{equation}
First we consider the case $i=\overline{q}$. We assume on the contrary that there is no such $r_F^{\overline{q}} \in \mathbb{T}_+$ for which the above property holds. Since $ \mathcal{U}(F)$ is not empty, by Proposition~\ref{string_prop} (i), there exists $r_0 \in \mathbb{T}_+$ such that
\begin{equation} \label{eqn2_up_lm1pf}
\mathcal{S}_r^{P^*} \cap \mathcal{U}(F) \neq \emptyset, ~\forall r \geq r_0.
\end{equation}
Thus, the non-existence of $r_F^{\overline{q}}$ implies that, for every $r\geq r_0$, there exists a string $\mathcal{R}_r \in \mathcal{U}(F)$ with len$(\mathcal{R}_r )\geq r$, such that $\overline{w}(\mathcal{R}_r^r)=0$. Therefore, such $\mathcal{R}_r$ is of the form
\begin{equation}
\mathcal{R}_r= \left(f_{2^N-1}^r,f_{\jmath_1},\cdots,f_{\jmath_{\mbox{\scriptsize len}( \mathcal{R}_r)-r}}, P^*\right)
\end{equation}
where ${\jmath_1},\cdots,{\jmath_{\mbox{\scriptsize len}( \mathcal{R}_r)-r}} \in \mathfrak{P}$. Thus, by denoting
\begin{equation}
X_r=f_{\jmath_1} \circ \cdots \circ f_{\jmath_{\mbox{\scriptsize len}( \mathcal{R}_r)-r}} (P^*),
\end{equation}
we have $\mathcal{N}(\mathcal{R}_r)=f^r_{2^N-1}(X_r)$.
By Proposition~\ref{unif_conv} (ii), the uniform convergence of the Riccati iterates implies that, for an arbitrary $\varepsilon >0$, there exists $r_{\varepsilon}\geq M $, such that, for every $X\in \mathbb{S}_+^M$,
\begin{equation}
\left\|f_{2^{N}-1}^{r}\left(X\right)-P^{\ast}\right\|\leq\varepsilon,~r\geq r_{\varepsilon}
\end{equation}
where the constant $r_{\varepsilon}$ can be chosen independently of $X$. Then, by defining $r'_{\varepsilon}=\max(r_0,r_{\varepsilon})$, we have
\begin{equation}
\left\| \mathcal{N}(\mathcal{R}_r) - P^{\ast} \right\|= \left\|f_{2^{N}-1}^{r}\left(X_r\right)-P^{\ast}\right\|\leq\varepsilon,~r\geq r'_{\varepsilon}.
\end{equation}
Since $\varepsilon$ is arbitrary, the above result shows that the sequence $\{ \mathcal{N} (\mathcal{R}_r)\}_{r\geq r'_{\varepsilon}}$ of numerical results converges to $P^*$ as $r \rightarrow \infty$. By construction, the sequence $\{ \mathcal{N} (\mathcal{R}_r)\}_{r\geq r'_{\varepsilon}}$ belongs to the set $F$, and we conclude that $P^*$ is a limit point of the set $F$. Since $F$ is closed, we have $P^*\in F$, which implies
\begin{equation}
\left\{ \mathcal{R} \in \mathcal{S}^{P^*} | \mathcal{N} (\mathcal{R})=P^* \right\} \subset \mathcal{U}(F).
\end{equation}
Hence, specifically, $\left(f_{2^{N}-1}, P^*\right) \in \mathcal{U}(F)$. Thus the fact that $\overline{w} \left((f_{2^{N}-1}, P^*) \right)=0$ contradicts the hypothesis $l'(F) \geq \overline{q}$.

Therefore, we establish that, if $l'(F) \geq \overline{q}$, there exists $r_F^{\overline{q}}$ satisfying the property in (\ref{eqn1_up_lm1pf}) for $i=\overline{q}$. Note here that, if $l'(F)=\overline{q}$, this step has completed the proof of the lemma. In the general case, to establish (\ref{eqn1_up_lm1pf}) for all $\overline{q} \leq i \leq l'(F)$, we need the following additional steps.

Let us now assume $l'(F) \geq 2\overline{q}$. We further assume on the contrary that the claim in (\ref{eqn1_up_lm1pf}) does not hold for any $\overline{q} \leq i \leq l'(F)$. By the previous step, clearly the claim holds for $i=\overline{q}$. Then, let $k$, $\overline{q} \leq k < l'(F)$, be the largest number such that the claim in (\ref{eqn1_up_lm1pf}) holds for all $\overline{q} \leq i \leq k$, which implies that there exists no $r_F^{k+\overline{q}} \in \mathbb{T}_+$ satisfying the claim in (\ref{eqn1_up_lm1pf}) for $i=k+\overline{q}$. Since the claim holds for $i=k$, there exists $r_F^k \in \mathbb{T}_+$ such that, for all $\mathcal{R} \in \mathcal{U}(F)$ with len$(\mathcal{R}) \geq r_F^k$, we have $\overline{w}(\mathcal{R}^{r_F^k} ) \geq k$. The non-existence of $r_F^{k+\overline{q}}$ and (\ref{eqn2_up_lm1pf}) imply that, for every $r \geq r_0$, there exists a string $\mathcal{R}_r \in \mathcal{U}(F)$ with len$(\mathcal{R}_r) \geq r$ such that $\overline{w}(\mathcal{R}_r^r) < k+\overline{q}$.

Define $r'_0=\max(r_0,r_F^k)$, then by the existence of $r_F^k$ and $\overline{w}(\mathcal{R}_r^r) < k+\overline{q}$, we have $ \overline{w}(\mathcal{R}_r^r) = k$ for $r\geq r'_0 $. Therefore, for $r\geq r'_0 $, $\mathcal{R}_r$ is necessarily of the form
\begin{equation}
\mathcal{R}_r= \left(f_{\jmath_1},\cdots,f_{\jmath_{r_F^k}}, f_{2^N-1}^{r-r_F^k},f_{i_1},\cdots,f_{i_{\mbox{\scriptsize len}( \mathcal{R}_r)-r}}, P^*\right)\nonumber
\end{equation}
where ${\jmath_1},\cdots,{\jmath_{r_f^k}} \in \mathfrak{P}$ such that $\overline{w}(\mathcal{R}_r^{r_F^k})=k$ and ${i_1},\cdots,{i_{\mbox{\scriptsize len}( \mathcal{R}_r)-r}} \in \mathfrak{P}$.

Now consider the sequence $\{ \mathcal{R}_r\}_{r\geq r'_0}$. Define the set $\mathcal{J}$ as $\mathcal{J}=\{ \mathcal{R}_r, ~r\geq r'_0 \}$, and also define the set $\mathcal{J}_1$ as $\mathcal{J}_1=\{ \mathcal{R} \in \mathcal{S}^{P^*}_{r_F^k}|\overline{w}(\mathcal{R} )=k\}$. Consider the mapping $\Theta^{r_F^k}: \mathcal{J} \mapsto \mathcal{J}_1$ by
\begin{equation}
\Theta^{r_F^k}(\mathcal{R})=\mathcal{R}^{r_F^k},~\forall \mathcal{R} \in \mathcal{J}.
\end{equation}

Since the cardinality of the set $\mathcal{J}_1$ is finite and the set $\mathcal{J}$ is countably infinite, for a specific $\mathcal{R}' \in \mathcal{J}_1$, the set $\left(\Theta^{r_F^k}\right)^{-1}(\mathcal{R}')$ is countably infinite. This in turn implies that we can extract a subsequence $\{ \mathcal{R}_{r_m}\}_{m \geq 0}$ from the sequence $\{ \mathcal{R}_{r}\}_{r \geq r'_0}$, such that
\begin{equation}
\mathcal{R}_{r_m}^{r_F^k} = \mathcal{R}', ~\forall m \geq 0.
\end{equation}
In other words, if $\mathcal{R}'$ is represented by $ \mathcal{R}'= \left( f_{\jmath'_1},\cdots,f_{\jmath'_{r_F^k}}, P^* \right)$ for some fixed $ {\jmath'_1},\cdots,{\jmath'_{r_F^k}} \in \mathfrak{P}$, for each $m$ the string $\mathcal{R}_{r_m}$ is of the form
\begin{equation}
\mathcal{R}_{r_m}= \left(f_{\jmath'_1},\cdots,f_{\jmath'_{r_F^k}}, f_{2^N-1}^{r_m-r_F^k},f_{i_1},\cdots,f_{i_{\mbox{\scriptsize len}( \mathcal{R}_{r_m})-r_m}}, P^*\right)\nonumber
\end{equation}
where ${i_1},\cdots,{i_{\mbox{\scriptsize len}( \mathcal{R}_{r_m})-r_m}} \in \mathfrak{P}$ are arbitrary. We denote by
\begin{equation}
X_m=f_{i_1} \circ \cdots \circ f_{i_{\mbox{\scriptsize len}( \mathcal{R}_{r_m})-r_m}} (P^*),~\forall m,
\end{equation}
and we have
\begin{equation}
\mathcal{N}(\mathcal{R}_{r_m})=f_{\jmath'_1}\circ \cdots \circ f_{\jmath'_{r_F^k}} \left( f_{2^N-1}^{r_m-r_F^k} (X_m) \right).
\end{equation}
Since $r_m \rightarrow \infty$ as $m \rightarrow \infty$, by Proposition 14 (ii), we have
\begin{equation}
\lim_{m\rightarrow \infty} f_{2^N-1}^{r_m-r_F^k} (X_m) =P^*.
\end{equation}
Note that the function $f_{\jmath'_1}\circ \cdots \circ f_{\jmath'_{r_F^k}}: \mathbb{S}^M_+ \mapsto \mathbb{S}^M_+ $, being the finite composition of continuous functions, is continuous. We then have \begin{align}
 \lim_{m\rightarrow \infty} \mathcal{N}(\mathcal{R}_{r_m}) &= \lim_{m\rightarrow \infty} f_{\jmath'_1}\circ \cdots \circ f_{\jmath'_{r_F^k}} \left( f_{2^N-1}^{r_m-r_F^k} (X_m) \right) \nonumber\\
&=   f_{\jmath'_1}\circ \cdots \circ f_{\jmath'_{r_F^k}} \left( \lim_{m\rightarrow \infty} f_{2^N-1}^{r_m-r_F^k} (X_m) \right) \nonumber\\
&=  f_{\jmath'_1}\circ \cdots \circ f_{\jmath'_{r_F^k}} \left( P^* \right) \nonumber\\
&=  \mathcal{N} (\mathcal{R}').
\end{align}
Therefore, the sequence $\{ \mathcal{N}(\mathcal{R}_{r_m})  \}_{m\geq 0}$ in $F$ converges to $ \mathcal{N} (\mathcal{R}') $ as $m \rightarrow \infty$. Hence $\mathcal{N} (\mathcal{R}')$ is a limit point in $F$, and $\mathcal{N} (\mathcal{R}') \in F$ as $F$ is closed. This implies that $\mathcal{R}' \in \mathcal{U}_F$. Since $\overline{w}(\mathcal{R}') =k$ and $\mathcal{R}' \in \mathcal{U}_F$, this contradicts the hypothesis that $k < l'(F)$ and thus the claim in (\ref{eqn1_up_lm1pf}) holds for all $\overline{q} \leq i \leq l'(F)$.

To prove, if $l(F)< \infty$, there
exists $r_F \in \mathbb{T}_+$ large enough, such that for all $\mathcal{R} \in \mathcal{U}(F)$
with len$(\mathcal{R}) \geq r_F$, we have $\pi(\mathcal{R}^{r_F}) \geq l(F)$, the method is the same as above, where $l(F)$ becomes a non-negative integer. We choose $r_F$ as the maximum one in these two cases, then the Lemma is proved.


\section{Proof of Lemma~\ref{lem_up_compact}} \label{proof_lem_up_compact}
For $\varepsilon > 0$, define $K_{\varepsilon}$ as the $\varepsilon$-neighborhood of $K$ and
$\overline{K_{\epsilon}}$ as its $\varepsilon$-closure, i.e.,
\begin{equation}
K_{\varepsilon}=\left\{X\in \mathbb{S}_+^{M} | \inf_{Y\in K} \left\| X-Y\right\| <\varepsilon \right\}
\end{equation}
\begin{equation}
\overline{K_{\varepsilon}}=\left\{X\in \mathbb{S}_+^{M} | \inf_{Y\in K} \left\| X-Y\right\| \leq \varepsilon \right\}.
\end{equation}

Since $K_{\varepsilon}$ is open, by the weak convergence of
the sequence $\{P_{n}(r)\}$ to $\mu^{\overline{\gamma}}$, we have
\begin{equation}
    \liminf_{r \rightarrow \infty} \mathbb{P} \left( P_{n}(r) \in K_{\varepsilon}
    \right) \geq \mu^{\overline{\gamma}}(K_{\varepsilon}),
\end{equation}
which implies that
\begin{equation}\label{ub_1}
     \liminf_{r \rightarrow \infty} \mathbb{P} \left( P_{n}(r) \in
     \overline{K_{\varepsilon}}
    \right) \geq \mu^{\overline{\gamma}}(K).
\end{equation}
Now we calculate the left-hand side of (\ref{ub_1}). Since $\overline{K_{\varepsilon}}$ is closed, the results of
Lemma~\ref{ub_lemma1} apply. Recall the definition of $\mathcal{U}(F)$. Also, for every $r\in \mathbb{T}_+$ and the closed set $F$, we define
\begin{equation}
     \mathcal{U}^r(F)=\mathcal{U}(F)\cap \mathcal{S}_r^{P^*}.
\end{equation}
We consider first $l(K)<\infty$ and $l'(K)<\infty$ (i.e., $\mathcal{U}(K)$ is non-empty). We then have
\begin{equation}\label{ub_2}
     \mathbb{P} \left( P_{n}(r) \in \overline{K_{\varepsilon}} \right)=
     \mathbb{P} \left( P_{n}(r) \in \mathcal{N}\left(\mathcal{U}^r\left(\overline{K_{\varepsilon}}\right)\right) \right).
\end{equation}
Since $K \subset \overline{K_{\varepsilon}}$ and $l(K) < \infty$, we have $l(\overline{K_{\varepsilon}})
< \infty$. Thus, since $\overline{K_{\varepsilon}}$ is closed, Lemma~\ref{ub_lemma1} shows that there
exists $r_{\overline{K_{\varepsilon}}} \in \mathbb{T}_+$, such that, for any string $\mathcal{R} \in
\mathcal{U}(\overline{K_{\varepsilon}})$ with len$(\mathcal{R}) \geq r_{\overline{K_{\varepsilon}}}$, we
have $\pi\left(\mathcal{R}^{r_{\overline{K_{\varepsilon}}}}\right) \geq l(\overline{K_{\varepsilon}})$ and $\overline{w}\left(\mathcal{R}^{r_{\overline{K_{\varepsilon}}}}\right) \geq l'(\overline{K_{\varepsilon}})$. In other words, for all $r\geq r_{\overline{K_{\varepsilon}}}$, we have
\begin{equation}
\pi\left(\mathcal{R}^{r_{\overline{K_{\varepsilon}}}}\right) \geq l(\overline{K_{\varepsilon}}),~\overline{w}\left(\mathcal{R}^{r_{\overline{K_{\varepsilon}}}}\right) \geq l'(\overline{K_{\varepsilon}}),~\forall~\mathcal{R} \in \mathcal{U}^r(\overline{K_{\varepsilon}}).
\end{equation}
Now consider $r\geq r_{\overline{K_{\varepsilon}}}$ and define $\mathcal{J}_r^{P^*}$ as the set of
strings
\begin{equation}
\mathcal{J}_r^{P^*}=\left\{\mathcal{R}\in \mathcal{S}_r^{P^*}|\pi
\left(\mathcal{R}^{r_{\overline{K_{\varepsilon}}}}\right) \geq l(\overline{K_{\varepsilon}}),~\overline{w}\left(\mathcal{R}^{r_{\overline{K_{\varepsilon}}}}\right) \geq l'(\overline{K_{\varepsilon}}) \right\}.
\end{equation}
The set $\mathcal{J}_r^{P^*}$ consists of all strings $\mathcal{R}$ with length $r$ such that we have at least $l(\overline{K_{\varepsilon}})$ occurrences of non-$f_{2^N-1}$ and the value of $\overline{w}$ no less than $l'(\overline{K_{\varepsilon}})$ in the truncated string $\mathcal{R}^{r_{\overline{K_{\varepsilon}}}}$.

For $r\geq r_{\overline{K_{\varepsilon}}}$, it is obvious that the following holds:
\begin{equation}\label{ub_3}
\mathcal{U}^r(\overline{K_{\varepsilon}}) \subset \mathcal{J}_r^{P^*} \subset \mathcal{S}_r^{P^*}.
\end{equation}
Clearly, we have that, for $r\geq r_{\overline{K_{\varepsilon}}}$,
\begin{align}\label{up_4}
&\mathbb{P} \left( P_{n}(r) \in \mathcal{N}(\mathcal{J}_r^{P^*})\right)\nonumber\\
&=\sum_{\mathcal{R}\in \mathcal{J}_r^{P^*}} \prod_{k=1}^{r} q_{n_k}(\jmath_k)
\leq (2^N-1)^{\underline{l}(\overline{K_{\varepsilon}})} \left(\begin{array}{l}
r_{\overline{K_{\varepsilon}}}\\
\underline{l}(\overline{K_{\varepsilon}})
\end{array}\right) \nonumber\\
&~~~~~~\times \max_{\mbox{\tiny$\begin{array}{c}
{(n_1,\cdots,n_{r_{\overline{K_{\varepsilon}}}})\in \mathcal{P}_{r_{\overline{K_{\varepsilon}}}} }\\
{\mathcal{R}:\pi
(\mathcal{R}^{r_{\overline{K_{\varepsilon}}}})=\underline{l}(\overline{K_{\varepsilon}})}\end{array}$}} \prod_{k=1}^{r_{\overline{K_{\varepsilon}}}} q_{n_k}(\jmath_k|{\jmath_k \neq 2^N-1}),
\end{align}
where $\underline{l}(\overline{K_{\varepsilon}})=\min_{\mathcal{R}\in \mathcal{J}_r^{P^*}} \pi
\left(\mathcal{R}^{r_{\overline{K_{\varepsilon}}}}\right)$. Then, from (\ref{ub_2}) and (\ref{ub_3}), we have
\begin{align}
    &\mu^{\overline{\gamma}}(K) \leq \liminf_{r\rightarrow \infty} \mathbb{P} \left( P_{n}(r) \in \mathcal{N}\left(\mathcal{U}^r(\overline{K_{\varepsilon}})\right)\right) \nonumber\\
     &\leq \liminf_{r\rightarrow \infty}\mathbb{P} \left( P_{n}(r) \in \mathcal{N}(\mathcal{J}_r^{P^*})\right) \nonumber\\
     & \leq (2^N-1)^{\underline{l}(\overline{K_{\varepsilon}})}\left(\begin{array}{l} r_{\overline{K_{\varepsilon}}}\nonumber\\
\underline{l}(\overline{K_{\varepsilon}})
\end{array}\right)\nonumber\\
&~~\max_{\mbox{\tiny$\begin{array}{c}
{(n_1,\cdots,n_{r_{\overline{K_{\varepsilon}}}})\in \mathcal{P}_{r_{\overline{K_{\varepsilon}}}} }\\
{\mathcal{R}:\pi
(\mathcal{R}^{r_{\overline{K_{\varepsilon}}}})=\underline{l}(\overline{K_{\varepsilon}})}\end{array}$}} \prod_{k=1}^{r_{\overline{K_{\varepsilon}}}}
 q_{n_k}(\jmath_k|{\jmath_k \neq 2^N-1}).\nonumber
\end{align}

Taking the logarithm, dividing by $\overline{\gamma}$ on both sides, and taking the limits, we have
\begin{align}
    &\limsup_{\overline{\gamma}\rightarrow \infty} \frac{\ln
    \mu^{\overline{\gamma}}(K)}{\overline{\gamma}}\nonumber\\
    &\leq \limsup_{\overline{\gamma}\rightarrow
    \infty} \max_{\mbox{\tiny$\begin{array}{c}
{(n_1,\cdots,n_{r_{\overline{K_{\varepsilon}}}})\in \mathcal{P}_{r_{\overline{K_{\varepsilon}}}} }\\
{\mathcal{R}:\pi
(\mathcal{R}^{r_{\overline{K_{\varepsilon}}}})=\underline{l}(\overline{K_{\varepsilon}})}\end{array}$}} \sum_{k=1}^{r_{\overline{K_{\varepsilon}}}}\mathbb{I}_{\jmath_k \neq 2^N-1}\frac{1}{\overline{\gamma}} \ln q_{n_k}(\jmath_k) \nonumber\\
    &\leq \max_{\mbox{\tiny$\begin{array}{c}
{(n_1,\cdots,n_{r_{\overline{K_{\varepsilon}}}})\in \mathcal{P}_{r_{\overline{K_{\varepsilon}}}} }\\
{\mathcal{R}:\pi
(\mathcal{R}^{r_{\overline{K_{\varepsilon}}}})=\underline{l}(\overline{K_{\varepsilon}})}\end{array}$}} \sum_{k=1}^{r_{\overline{K_{\varepsilon}}}}\mathbb{I}_{\jmath_k \neq 2^N-1}
    \left( \limsup_{\overline{\gamma}\rightarrow
    \infty} \frac{1}{\overline{\gamma}}
    \ln q_{n_k}(\jmath_k)  \right)\nonumber\\
&\leq  - \min_{\mbox{\tiny$\begin{array}{c}
{(n_1,\cdots,n_{r_{\overline{K_{\varepsilon}}}})\in \mathcal{P}_{r_{\overline{K_{\varepsilon}}}} }\\
{\mathcal{R}:\pi
(\mathcal{R}^{r_{\overline{K_{\varepsilon}}}})=\underline{l}(\overline{K_{\varepsilon}})}\end{array}$}} \sum_{k=1}^{r_{\overline{K_{\varepsilon}}}}\mathbb{I}_{\jmath_k \neq 2^N-1}
\overline{q}_{n_k}(\jmath_k) \leq -l^{'}(\overline{K_{\varepsilon}}).\nonumber
\end{align}

Then, taking the limit as $\varepsilon \rightarrow 0$ on both sides leads to
\begin{equation} \label{up_last2}
\limsup_{\overline{\gamma}\rightarrow \infty} \frac{\ln
    \mu^{\overline{\gamma}}(K)}{\overline{\gamma}} \leq -\lim_{\varepsilon \rightarrow
    0}l'(\overline{K_{\varepsilon}}).
\end{equation}
From Proposition~\ref{prop_rate_function} (iii), we have
\begin{equation}
   l'(\overline{K_{\varepsilon}}) =\inf_{X \in {\overline{K_{\varepsilon}}} } \overline{I}(X)  \geq \inf_{X \in {\overline{K_{\varepsilon}}} } \overline{I}_L(X)
\end{equation}
where $\overline{I}(X)=\inf_{\mathcal{R}\in \mathcal{S}^{P^*}(X)}\overline{w}(\mathcal{R})$.

Again, taking the limit as $\varepsilon \rightarrow 0$ and from Proposition~\ref{prop_rate_function} (iv), we have
\begin{equation}\label{up_last1}
   \lim_{\varepsilon \rightarrow
    0}l'(\overline{K_{\varepsilon}}) \geq \lim_{\varepsilon \rightarrow
    0} \inf_{X \in {\overline{K_{\varepsilon}}} } \overline{I}_L(X)= \inf_{X \in {K }} \overline{I}_L (X).
\end{equation}
The lemma then follows from (\ref{up_last2}) and (\ref{up_last1}).

\section{Proof of Lemma~\ref{lem_up_tightness}} \label{proof_lem_up_tightness}
Let $a>0$ be arbitrary and choose $z \in \mathbb{N}$ such that $z\geq a$. From Proposition~\ref{string_prop} (ii), there exists
$\alpha_{P^{*}}\in\mathbb{R}_{+}$ depending on $P^{*}$ only, such
that
\begin{equation}
\label{concatS41}
f_{\hat{\jmath}_{\pi(\mathcal{R})}}\circ f_{\hat{\jmath}_{\pi(\mathcal{R})-1}}\cdots\circ f_{\hat{\jmath}_{1}}\left(\alpha_{P^{*}}I\right)\succeq\mathcal{N}\left(\mathcal{R}\right),~\forall \mathcal{R}\in \mathcal{S}^{P^*}.\nonumber
\end{equation}
 We define $b\in \mathbb{R}_+$ such that $\|f_{\hat{\jmath}_{z}}\circ f_{\hat{\jmath}_{z-1}}\cdots\circ f_{\hat{\jmath}_{1}}\left(\alpha_{P^{*}}I\right)\| <b$. Consider the compact set $K_a=\{X \in \mathbb{S}_+^M |\|X\|\leq b\}$, and also define the closed set $F_b=\{X \in \mathbb{S}_+^M |\|X\|\geq b\}$. From Lemma~\ref{ub_lemma1}, define the set $\mathcal{U}(F_b)$ as
\begin{equation}
\mathcal{U}(F_b)=\left\{\mathcal{R}\in \mathcal{S}^{P^*}|\mathcal{N}(\mathcal{R})\in F_b   \right\}.
\end{equation}
Then, we have the following inclusion:
\begin{equation}
\mathcal{U}(F_b)\subset \left\{\mathcal{R}\in \mathcal{S}^{P^*}|\pi(\mathcal{R})\geq z  \right\}.
\end{equation}
Hence, $l(F_b)=\inf_{\mathcal{R} \in \mathcal{U}(F_b)} \pi(\mathcal{R}) \geq z$.
Since $F_b$ is closed, by Lemma~\ref{ub_lemma1}, there exists $r_{F_b} \in \mathbb{T}_+$ such that
\begin{equation}
\pi(\mathcal{R}^{r_{F_b}}) \geq z,~\forall \mathcal{R}\in \mathcal{U}(F_b).
\end{equation}
To estimate the probability $\mu^{\overline{\gamma}}(K_a^C)$, we follow the method in Lemma~\ref{lem_up_compact}. First, we have the following by weak convergence:
\begin{equation}
\mu^{\overline{\gamma}}(K_a^C) \!\leq\! \liminf_{r \rightarrow \infty} \mathbb{P}\!\left( P_{n}(r) \!\in\! K_a^C \right)\! \leq\!  \liminf_{r \rightarrow \infty} \mathbb{P}\!\left( P_{n}(r)\!\in\! F_b \right).\nonumber
\end{equation}
For $r\in \mathbb{T}_+$, denote the set $\mathcal{J}_r^{P^*}=\mathcal{S}_r^{P^*} \cap \mathcal{U}(F_b).$ For $r\geq r_{F_b}$, similar to (\ref{up_4}), we have
 \begin{align}
&\mathbb{P} \left( {P_{n}(r)} \in F_b\right)=\sum_{\mathcal{R}\in \mathcal{J}_r^{P^*}} \prod_{k=1}^{r} q_{n_k}(\jmath_k)\nonumber\\
&\leq(2^N-1)^{\underline{l}(F_b)}\left(\begin{array}{l} r_{F_b}\nonumber\\
\underline{l}(F_b)
\end{array}\right)\nonumber\\
&~~\max_{\mbox{\tiny$\begin{array}{c}
{(n_1,\cdots,n_{r_{F_b}})\in \mathcal{P}_{r_{F_b}} }\\
{\mathcal{R}:\pi
(\mathcal{R}^{r_{F_b}})=\underline{l}(F_b)}\end{array}$}} \prod_{k=1}^{r_{F_b}}
 q_{n_k}(\jmath_k|{\jmath_k \neq 2^N-1})\nonumber\\
 &\leq(2^N-1)^{\underline{l}(F_b)}\left(\begin{array}{l} r_{F_b}\nonumber\\
\underline{l}(F_b)
\end{array}\right)\nonumber\\
&~~\max_{\mbox{\tiny$\begin{array}{c}
{(n_1,\cdots,n_{r_{z}})\in \mathcal{P}_{r_{z}} }\\
{\mathcal{R}:\pi
(\mathcal{R}^{r_{z}})=z}\end{array}$}} \prod_{k=1}^{r_{z}}
 q_{n_k}(\jmath_k|{\jmath_k \neq 2^N-1}).\nonumber
\end{align}
Arguments similar to those in Lemma~\ref{lem_up_compact} lead to
\begin{align}
&\mu^{\overline{\gamma}}(K_a^C) \leq(2^N-1)^{\underline{l}(F_b)}\left(\begin{array}{l} r_{F_b}\nonumber\\
\underline{l}(F_b)
\end{array}\right) \nonumber\\
&\max_{\mbox{\tiny$\begin{array}{c}
{(n_1,\cdots,n_{r_{z}})\in \mathcal{P}_{r_{z}} }\\
{\mathcal{R}:\pi
(\mathcal{R}^{r_{z}})=z}\end{array}$}} \prod_{k=1}^{r_{z}}
 q_{n_k}(\jmath_k|{\jmath_k \neq 2^N-1}),\nonumber
\end{align}
from which we obtain,
\begin{align}
 &\limsup_{\overline{\gamma}\rightarrow \infty} \frac{\ln\mu^{\overline{\gamma}}(K_a^C)}{\overline{\gamma}} \leq  \nonumber\\
    &-\min_{\mbox{\tiny$\begin{array}{c}
{(n_1,\cdots,n_{r_{z}})\in \mathcal{P}_{r_{z}} }\\
{\mathcal{R}:\pi
(\mathcal{R}^{r_{z}})=z}\end{array}$}} \sum_{i=1}^{r_z}\mathbb{I}_{\jmath_{i}\neq 2^{N}-1}\overline{q}_{n_{i}}(\jmath_{i})\leq -W(z),\nonumber
\end{align}
where $W(z)$ is defined as
\begin{align}
W(z)=\min_{\mbox{\tiny$\begin{array}{c}
{(n_1,\cdots,n_{{\mbox{\scriptsize len}(\mathcal{R})}})\in \mathcal{P}_{{\mbox{\scriptsize len}(\mathcal{R})}} }\\
{\mathcal{R}:\pi
(\mathcal{R})=z}\end{array}$}} \sum_{i=1}^{{\mbox{\scriptsize len}(\mathcal{R})}}\mathbb{I}_{\jmath_{i}\neq 2^{N}-1}\overline{q}_{n_{i}}(\jmath_{i}).
\end{align}
Obviously, $W(z)\geq W(a)$ follows from $z\geq \lfloor a \rfloor$. Then we have $\limsup_{\overline{\gamma}\rightarrow \infty} \frac{\ln
    \mu^{\overline{\gamma}}(K_a^C)}{\overline{\gamma}} \leq -W(a)$.

\bibliographystyle{IEEEtran}
\bibliography{IEEEabrv,CentralBib}
\begin{IEEEbiographynophoto}{Di Li}(S'13) received the B.Eng. degree in Automation Engineering and the M.S. degree in Information and Communication Engineering from Beijing University of Posts and Telecommunications, Beijing, China, in 2008 and 2011, respectively. He is currently pursuing a Ph.D. degree in Electrical and Computer
 Engineering at Texas A\&M University, College Station, TX. His research interests include statistical signal processing, distributed estimation and detection, and change-point detection.
\end{IEEEbiographynophoto}
\begin{IEEEbiographynophoto}{Soummya Kar} (S'05--M'10) received
 the B.Tech. degree in Electronics and Electrical
 Communication Engineering from the Indian Institute
 of Technology, Kharagpur, India, in May 2005
 and the Ph.D. degree in electrical and computer
 engineering from Carnegie Mellon University,
 Pittsburgh, PA, in 2010. From June 2010 to May 2011 he was with the EE Department at Princeton University as a Postdoctoral Research Associate.

 He is currently an Assistant Research Professor of ECE at Carnegie Mellon University.

 His research interests include performance
 analysis and inference in large-scale networked systems, adaptive stochastic
 systems, stochastic approximation, and large deviations.
\end{IEEEbiographynophoto}

\begin{IEEEbiographynophoto}{Jos\'e M.~F.~Moura}(S'71--M'75--SM'90--F'94) is the Philip L.~and Marsha Dowd University Professor at Carnegie Mellon University (CMU). He received the engenheiro electrot\'{e}cnico degree from Instituto Superior T\'ecnico (IST), Lisbon, Portugal, and the M.Sc., E.E., and D.Sc.~degrees in EECS from MIT, Cambridge, MA. He was on the faculty at IST and a visiting Professor at MIT and NYU. He is founding director of ICTI@CMU, a large education and research program between CMU and Portugal, www.cmuportugal.org. His research interests are on data science and include statistical, algebraic, and distributed signal processing on graphs. He has published over 470 papers, has eleven patents issued by the US Patent Office, and cofounded SpiralGen.

Dr.~Moura was elected 2016~IEEE Vice-President of Technical Activities. He served as IEEE Division~IX Director and IEEE Board Director and on several IEEE Boards. He was \emph{President} of the \emph{IEEE Signal Processing Society}(SPS), served as \emph{Editor in Chief} for the {\em IEEE Transactions in SP}, interim \emph{Editor in Chief} for the \emph{IEEE SP Letters}, and member of several Editorial Boards, including \emph{IEEE Proceedings}, \emph{IEEE SP Magazine}, and the ACM \emph{Transactions on Sensor Networks}.

Dr.~Moura is member of the \textit{US National Academy of Engineering}, corresponding member of the \textit{Academy of Sciences of Portugal}, \emph{Fellow} of the \emph{IEEE}, and \emph{Fellow} of the \textit{AAAS}. He received the IEEE Signal Processing Society \textit{Technical Achievement Award} and the IEEE Signal Processing Society \textit{Society Award}.
%
\end{IEEEbiographynophoto}
\begin{IEEEbiographynophoto}{H. Vincent Poor} (S'72--M'77--SM'82--F'87) received the Ph.D. degree in electrical engineering and computer science from Princeton University in 1977.  From 1977 until 1990, he was on the faculty of the University of Illinois at Urbana-Champaign. Since 1990 he has been on the faculty at Princeton, where he is the Dean of Engineering and Applied Science, and the Michael Henry Strater University Professor of Electrical Engineering. He has also held visiting appointments at several other institutions, most recently at Imperial College and Stanford. Dr. Poor's research interests are in the areas of information theory, stochastic analysis and statistical signal processing, and their applications in wireless networks and related fields. Among his publications in these areas is the recent book \emph{Mechanisms and Games for Dynamic Spectrum Allocation} (Cambridge University Press, 2014).

Dr. Poor is a member of the National Academy of Engineering and the National Academy of Sciences, and is a foreign member of Academia Europaea and the Royal Society. He is also a fellow of the American Academy of Arts and Sciences, the Royal Academy of Engineering (U. K.), and the Royal Society of Edinburgh.  In 1990, he served as President of the IEEE Information Theory Society, in 2004-07 as the Editor-in-Chief of these TRANSACTIONS, and in 2009 as General Co-chair of the IEEE International Symposium on Information Theory, held in Seoul, South Korea. He received a Guggenheim Fellowship in 2002 and the IEEE Education Medal in 2005. Recent recognition of his work includes the 2014 URSI Booker Gold Medal, and honorary doctorates from several universities in Asia and Europe.
\end{IEEEbiographynophoto}

\begin{IEEEbiographynophoto}{Shuguang Cui} (S'99-M'05-SM'12-F'14) received his Ph.D. in Electrical Engineering from Stanford University, California, USA, in 2005, M.Eng. in Electrical Engineering from McMaster University, Hamilton, Canada, in 2000, and B.Eng. in Radio Engineering with the highest distinction (ranked No.1 in the department) from Beijing University of Posts and Telecommunications, Beijing, China, in 1997. He has been working as an associate professor in Electrical and Computer Engineering at the Texas A\&M University, College Station, TX. His current research interests focus on data oriented large-scale information analysis and system design, including large-scale distributed estimation and detection, information theoretical approaches for large data set analysis, and complex cyber-physical system design.

His research papers have been highly cited; according to the data on 2/16/2014 from Web of Science, 8 of them had been ranked within the top 10 most highly cited papers (one of them ranked No.1 and three of them ranked No.2) among all published over the same periods in the corresponding journals. He was selected as the Thomson Reuters Highly Cited Researcher in 2014. He was the recipient of the {\textbf{IEEE Signal Processing Society 2012 Best Paper Award}}. He has been serving as the TPC co-chairs for many IEEE conferences. He has also been serving as the associate editors for IEEE Transactions on Signal Processing and IEEE Transactions on Wireless Communications. He is the elected member for IEEE Signal Processing Society SPCOM Technical Committee (2009~2015) and the elected Secretary for IEEE ComSoc Wireless Technical Committee. He is the member of the Steering Committee for the new IEEE Transactions on Big Data, in charge of identifying the first EiC. He was elected as an IEEE Fellow in 2013.
\end{IEEEbiographynophoto}
\end{document}